\documentclass[a4paper,UKenglish,cleveref, autoref, thm-restate]{lipics-v2021}

\pdfoutput=1 
\hideLIPIcs  

\title{Finding and Counting Patterns in Sparse Graphs}

\author{Balagopal Komarath}{IIT Gandhinagar, India}{bkomarath@rbgo.in}{}{}

\author{Anant Kumar}{IIT Gandhinagar, India}{kumar_anant@iitgn.ac.in}{}{}

\author{Suchismita Mishra}{Universidad Andr\'es Bello, Chile}{suchismita.m@iitgn.ac.in}{}{}

\author{Aditi Sethia}{IIT Gandhinagar, India}{aditi.sethia@iitgn.ac.in}{}{}

\authorrunning{B. Komarath, A. Kumar, S. Mishra, and A. Sethia} 

\Copyright{Jane Open Access and Joan R. Public} 



\category{} 

\relatedversion{} 

\funding{}

\acknowledgements{}

\nolinenumbers 

\usepackage[svgnames]{xcolor}

\usepackage{etoolbox}
\makeatletter
\patchcmd{\BR@backref}{\newblock}{\newblock($\uparrow$~}{}{}
\patchcmd{\BR@backref}{\par}{)\par}{}{}
\makeatother

\hypersetup{
    colorlinks=true,
    linkcolor=IndianRed,
    filecolor=magenta,      
    urlcolor=DodgerBlue,
    citecolor = SeaGreen
}

\usepackage{charter}
\usepackage{parskip}

\usepackage{wrapfig}

\usepackage{tikz}
\usetikzlibrary{positioning}
\usetikzlibrary{arrows}
\usetikzlibrary{decorations}
\usetikzlibrary{shapes.geometric}
\usepackage{algorithm}
\usepackage{hyperref}
\usepackage{algpseudocode}  
\usepackage{latexsym,amssymb,amsmath,mathtools,eulervm,xspace,nicefrac}
\usepackage[color=blue!20,textsize=small,textwidth=2.2cm]{todonotes}

\newcommand{\NP}{\ensuremath{\mathsf{NP}}\xspace}

\usepackage{tikzit}

\tikzstyle{vertex}=[fill=black, draw=black, shape=circle, tikzit fill=black, tikzit shape=circle, tikzit draw=black]
\tikzstyle{bag}=[fill=none, draw=black, shape=rectangle, tikzit shape=circle]
\tikzstyle{comment}=[fill=white, draw=none, shape=rectangle, tikzit shape=rectangle]
\tikzstyle{etc}=[dashed]
\tikzstyle{gate}=[fill=none, draw=none, shape=circle]

\tikzstyle{wire}=[->]

\tikzstyle{v}=[fill=black, shape=circle]

\tikzstyle{red}=[-, draw=red]
\tikzstyle{rl}=[-, draw=red, bend left]
\tikzstyle{rr}=[-, draw=red, bend right]

\newcommand{\gaut}{\mathsf{Aut}}
\newcommand{\ghom}{\mathsf{Hom}}
\newcommand{\gsub}{\mathsf{Sub}}
\newcommand{\gind}{\mathsf{Ind}}
\newcommand{\edgemon}{\mathsf{EdgeMon}}
\newcommand{\vertexmon}{\mathsf{VertexMon}}
\newcommand{\evmon}{\mathsf{Mon}}
\newcommand{\mapgate}{\mathsf{MapGate}}
\newcommand{\spasm}{\mathsf{Spasm}}
\newcommand{\merge}{\mathsf{Merge}}
\newcommand{\trace}{\mathsf{trace}}
\newcommand{\td}{\mathsf{td}}
\newcommand{\mtd}{\mathsf{mtd}}
\newcommand{\tw}{\mathsf{tw}}
\newcommand{\mtw}{\mathsf{mtw}}
\newcommand{\ghw}{\mathsf{ghw}}
\newcommand{\otilde}{\widetilde{O}}
\newcommand{\depth}{\mathsf{depth}}
\newcommand{\dist}{\mathsf{dist}}
\newcommand{\graphw}{T_{3,3}}
\DeclareMathOperator{\pluseq}{\mathbin{{+}{=}}}

\EventEditors{John Q. Open and Joan R. Access}
\EventNoEds{2}
\EventLongTitle{42nd Conference on Very Important Topics (CVIT 2016)}
\EventShortTitle{CVIT 2016}
\EventAcronym{CVIT}
\EventYear{2016}
\EventDate{December 24--27, 2016}
\EventLocation{Little Whinging, United Kingdom}
\EventLogo{}
\SeriesVolume{42}
\ArticleNo{23}

\begin{document}

\maketitle

\begin{abstract}
    We consider algorithms for finding and counting small, fixed graphs in sparse host graphs. In the non-sparse setting, the parameters treedepth and treewidth play a crucial role in fast, constant-space and polynomial-space algorithms respectively. We discover two new parameters that we call matched treedepth and matched treewidth. We show that finding and counting patterns with low matched treedepth and low matched treewidth can be done asymptotically faster than the existing algorithms when the host graphs are sparse for many patterns. As an application to finding and counting fixed-size patterns, we discover $\otilde(m^3)$-time \footnote{$\otilde$ hides factors that are logarithmic in the input size.}, constant-space algorithms for cycles of length at most $11$ and $\otilde(m^2)$-time, polynomial-space algorithms for paths of length at most $10$.
\end{abstract}

\section{Introduction}
\label{sec:introduction}

Given simple graphs $G$, called the \emph{pattern}, and $H$, called the \emph{host}, a fundamental computational problem is to find or count occurrences of $G$ in $H$. What does it mean for $G$ to occur in $H$? The three most common notions of occurrence are characterized by mappings $\phi: V(G) \mapsto V(H)$. We say:
\begin{enumerate}
    \item If $\{u, v\} \in E(G)$ implies $\{\phi(u), \phi(v)\} \in E(H)$ and $\phi$ is one-to-one, then we say that $\phi$ witnesses a subgraph isomorphic to $G$ in $H$. The subgraph is obtained by taking the vertices and edges in the image of $\phi$. The number of $G$-subgraphs of $H$ is just the number of such subgraphs $G'$ of $H$.
    \item If $\{u, v\} \in E(G)$ is equivalent to $\{\phi(u), \phi(v)\} \in E(H)$ and $\phi$ is one-to-one, then $\phi$ witnesses an induced subgraph isomorphic to $G$ in $H$. The induced subgraph is obtained by taking the vertices and \emph{all} edges induced by those vertices in the image of $\phi$.
    \item  If $\{u, v\} \in E(G)$ implies $\{\phi(u), \phi(v)\} \in E(H)$, then we say that $\phi$ is a homomorphism from $G$ to $H$. Note that unlike a subgraph isomorphism, $\phi$ is not required to be one-to-one.
\end{enumerate}

For any of these notions, the detection problem is clearly in $\NP$. All three of them are also straightforward generalizations of the $\NP$-hard problem $\mathsf{CLIQUE}$. Therefore, the existence of efficient algorithms for finding or counting patterns under any of these notions is unlikely in general.

The class of pattern detection and counting problems remain interesting even if we restrict our attention to fixed pattern graphs. Williams~\cite{Williams2008} showed that the improved algorithms for finding triangles could be used to find faster algorithms for even $\NP$-complete problems such as $\mathsf{MAX2SAT}$. For fixed pattern graphs of size $k$, the brute-force search algorithm is as follows: Iterate over all $k$-tuples over $V(H)$ and check whether $G$ occurs in the induced subgraph of $H$ on the vertices in that $k$-tuple. This algorithm takes $\theta(n^k)$ time and constant space. Therefore, when we restrict our attention to fixed patterns, we seek improvements over this running time preferably keeping the space usage low. There are two broad techniques that reduce the running-time: the usage of fast matrix multiplication algorithms as a sub-routine and the exploitation of structural properties of pattern graphs.

If $A$ is the adjacency matrix of the graph, then Ne\v set\v ril and Poljak \cite{N1985} showed that one can obtain an $O(n^\omega)$-time algorithm for counting triangles using the identity $\trace(A^3) = 6\Delta$, where $\Delta$ is the number of triangles in the graph, where $\omega < 2.38$ is the matrix multiplication exponent. Using a simple reduction, they extended this to an algorithm to count $3k$-cliques in $O(n^{k\omega})$-time. They also showed that we can use improved algorithms for counting $k$-cliques to count any $k$-vertex pattern. Later, Kloks, Kratsch and M\"uller \cite{DBLP:journals/ipl/KloksKM00} showed how to use fast rectangular matrix multiplication to obtain similar improvements to the running time for counting cliques of all sizes, not just multiples of three. Note that the improvements obtained by these algorithms are applicable to all $k$-vertex patterns. i.e., they do not use the pattern's structure to obtain better algorithms. Since finding a $k$-clique requires $n^{\Omega(k)}$-time unless ETH is false, we need to exploit the structure of the pattern to obtain significantly better algorithms.

For patterns sparser than cliques, the run-time can be significantly improved over even fast matrix multiplication based (pattern finding) algorithms. The crucial idea is to exploit the structure of the pattern graph. A \emph{$k$-walk polynomial} is a polynomial where the monomials correspond to walks that are $k$ vertices long. For example, a walk $(u, v, w, x)$ will correspond to the monomial $x_{uv}x_{vw}x_{wx}$ and a walk $(u, v, u, v)$ to the monomial $x_{uv}^3$ \footnote{We write $uv$ to denote the edge $\{u, v\}$.}. Williams~\cite{10.1016/j.ipl.2008.11.004} showed that we can detect $k$-paths in graphs by (1) computing the $k$-walk polynomial and (2) checking whether it has multilinear monomials. We can compute the $k$-walk polynomial in linear-time using a simple dynamic programming algorithm and then multilinear monomials can be detected with high probability by evaluating this polynomial over an appropriate ring where the randomly chosen elements satisfy $a^2 = 0$. This yields is a $O(2^k(n+m))$-time algorithm for finding $k$-vertex paths as subgraphs in $n$-vertex, $m$-edge host graphs.
 
We now consider the problem of counting sparse patterns. For counting $k$-paths as subgraphs, the best known algorithm by Curticapean, Dell and Marx \cite{10.1145/3055399.3055502} takes only $O(f(k) n^{0.174k + o(k)})$-time for some function $f$. Coming to fixed pattern graphs, Alon, Yuster and Zwick \cite{Alon1997} gave $O(n^\omega)$-time algorithms for counting cycle subgraphs of length at most $8$ using an algorithm that \emph{combines} fast matrix multiplication and exploitation of the structure of the pattern. Notice that this is the same as the time required for counting triangles ($3$-cliques).
 
The notion of graph homomorphisms was shown to play a crucial role in all the above improved algorithms for finding and counting non-clique subgraphs. More specifically, Fomin, Lokshtanov, Raman, Rao, and Saurabh \cite{DBLP:journals/jcss/FominLRSR12} showed how the efficient construction of homomorphism polynomials (see Definition~\ref{def:hompoly}), a generalization of $k$-walk polynomials, can be used to detect subgraphs with small treewidth efficiently. Their algorithm can be seen as a generalization of Williams's algorithm \cite{10.1016/j.ipl.2008.11.004} for $k$-paths to arbitrary graphs. Similarly, Curticapean, Dell and Marx \cite{10.1145/3055399.3055502} showed that efficient algorithms for counting subgraphs can be derived from efficient algorithms for counting homomorphisms of graphs of small treewidth. Their algorithm can be seen as a generalization of the cycle-counting algorithms of Alon, Yuster, and Zwick \cite{Alon1997}.

Algorithms for finding and counting patterns in \emph{sparse} host graphs are also studied. An additional parameter, $m$, the number of edges in the host graph, is taken into account for the design and analysis of these algorithms. In the worst-case, $m$ could be as high as $\binom{n}{2}$, and hence, an $O(n^t)$-time algorithm and an $O(m^{t/2})$-time algorithm for some $t$ have the same asymptotic time complexity. However, it is common in practice that $m = o(n^2)$. For example, if the host graph models a road network, then $m = O(n)$, where the constant factor is determined by the maximum number of roads at any intersection. In such cases, an $O(m^{t/2})$-time algorithm is asymptotically better than an $O(n^t)$-time algorithm.

The broad themes of using fast matrix multiplication and/or structural parameters of the pattern to obtain improved algorithms are still applicable in the setting of sparse host graphs. Using fast matrix multiplication, Eisenbrand and Grandoni \cite{EisenbrandG04} showed that we can count $k$-cliques in $O(m^{k\omega/6})$-time. Kloks, Kratsch and M\"uller \cite{DBLP:journals/ipl/KloksKM00} showed that $K_4$ subgraphs can be counted in $O(m^{(\omega+1)/2})$-time. Again, since $\omega < 3$, this is better than the $O(m^2)$-time given by the brute-force algorithm. Using structural parameters of the pattern, Kowaluk, Lingas, and Lundell \cite{Kowaluk2013CountingAD} obtained many improved algorithms in the sparse host graph setting. For example, their methods obtain an algorithm that runs in $O(m^4)$-time for counting $P_{10}$ as subgraphs. In this work, we obtain an $\otilde(m^2)$-time algorithm for counting $P_{10}$ (See Theorems~\ref{thm:pkcountmtw},\ref{thm:ckcountmtw},\ref{thm:csixind},\ref{thm:pkbarind} for similar improvements).

The model of computation that we consider is the unit-cost RAM model. In particular, we can store labels of vertices and edges in the host graph in a constant number of words\footnote{In the TM model or the log-cost RAM model, storing labels of vertices would take $O(\log n)$ space.}. In this model, algorithms based on fast matrix multiplication and/or treewidth mentioned above use polynomial space. However, the brute-force search algorithm uses only constant space as it only needs to store $k$ vertex labels at a time (Recall that we regard $k$ as a constant.). How much speed-up can we obtain while preserving constant space usage? The graph parameter \emph{treedepth} plays a crucial role in answering this question. It is well known that we can count the homomorphisms from a pattern of treedepth $d$ in $O(n^d)$-time while using only constant space (See Komarath, Rahul, and Pandey \cite{DBLP:journals/corr/abs-2011-04778} for a construction of arithmetic formulas counting them. These arithmetic formulas can be implicitly constructed and evaluated in constant space.). Since all $k$-vertex patterns except $k$-clique has treedepth strictly less than $k$, this immediately yields an improvement over the running-time of brute-force while preserving constant space usage. In this work, we improve upon the treedepth-based algorithms for sparse host graphs where the pattern graph is a cycle of length at most $11$ (See Theorem~\ref{thm:ckcountmtd}).

\subsection{Connection to arithmetic circuits for graph homomorphism polynomials}

A popular sub-routine in these algorithms is an algorithm by Diaz, Serna and Thilikos \cite{10.5555/646719.702252} that efficiently counts the number of homomorphisms from a pattern of small treewidth to an arbitrary host graph. Indeed, it can be shown that this algorithm can be easily generalized to \emph{efficiently construct} circuits for homomorphism polynomials instead of counting homomorphisms. Bl\"aser, Komarath and Sreenivasaiah \cite{DBLP:conf/fsttcs/BlaserKS18} showed that efficient constructions for homomorphism polynomials can even be used to detect \emph{induced} subgraphs in some cases. They also show that many of the faster induced subgraph detection algorithms, such finding four-node subgraphs by Williams et al.\ \cite{10.5555/2722129.2722240} and five-node subgraphs by Kowaluk, Lingas, and Lundell \cite{Kowaluk2013CountingAD} can be described as algorithms that efficiently construct these homomorphism polynomials. Therefore, arithmetic circuits for graph homomorphism polynomials provide a unifying framework for describing almost all the fast algorithms that we know for finding and counting subgraphs and finding induced subgraphs. Can we improve these algorithms by finding more efficient ways to construct arithmetic circuits for homomorphism polynomials? Unfortunately, it is known that for the type of circuit that is constructed, i.e., circuits that do not involve cancellations, the existing constructions are the best possible for \emph{all} pattern graphs, as shown by Komarath, Pandey, and Rahul \cite{DBLP:journals/corr/abs-2011-04778}. The situation is similar for constant space algorithms. The best known algorithms can be expressed as divide-and-conquer algorithms that evaluate small formulas constructed by making use of the graph parameter treedepth. Komarath, Pandey, and Rahul\cite{DBLP:journals/corr/abs-2011-04778} also showed that the running-time of these algorithms match the best possible formula size for \emph{all} pattern graphs. These arithmetic circuit lower bounds serve as a technical motivation for considering sparse host graphs, in addition to the practical motivation mentioned earlier.

\subsection{Our findings}
In this paper, we study algorithms for finding and counting patterns in host graphs that work well especially when the host is sparse. We discover algorithms that are (1) strictly better than the brute-force algorithm, (2) strictly better than the best-known algorithms when the host graph is sparse, (3) close to the best-known algorithms when the host graphs are dense. Our algorithms are based on two new structural graph parameters -- the \emph{matched treedepth} and \emph{matched treewidth}. (See \ref{def:matchedET} and \ref{def:matched-tree-decomposition-treewidth} for formal definitions). We show that they can be used to obtain improved running times for algorithms that use constant space and polynomial space respectively. Our algorithms are summarized in Table~\ref{tab:algo}. In the table, the parameter $m$ is the number of edges in the host graph. We denote using $\mtw$ the matched treewidth of the pattern and using $\mtd$ the matched treedepth of the pattern. The notation $\otilde$ hides factors that are poly-logarithmic in the input (the host graph) size.

\begin{table}[ht!]
    \centering
    \begin{tabular}{llllll}
         Pattern & Type & Problem & Time & Space & Remarks \\
         \hline
         $C_k$ & Subgraph & Counting & $\otilde(m^3)$ & $O(1)$ & $k \leq 11$ \\
         $P_k$ & Subgraph & Counting & $\otilde(m^2)$ & $\otilde(m^2)$ & $k \leq 10$ \\
         $C_k$ & Subgraph & Counting & $\otilde(m^2)$ & $\otilde(m^2)$ & $k \leq 9$ \\
         Any & Homomorphism & Counting & $\otilde(m^{\lceil \mtd/2 \rceil})$ & $O(1)$ &  \\
         Any & Homomorphism & Counting & $\otilde(m^{\lceil (\mtw+1)/2 \rceil})$ & $\otilde(m^{\lceil (\mtw+1)/2 \rceil})$ &  \\
         $C_6$ & Induced subgraph & Detection & $\otilde(m^2)$ & $\otilde(m^2)$ & \\
         $\overline{P_k}$ & Induced subgraph & Detection & $\otilde(m^{\lceil (k-2)/2 \rceil})$ &  $\otilde(m^{\lceil (k-2)/2 \rceil})$ & \\
         \hline
    \end{tabular}
    \caption{Pattern counting and detection algorithms for sparse host graphs.}
    \label{tab:algo}
\end{table}

We now explain the relevance of our new parameters; state our algorithms, the relationships between various graph parameters, and some structural characterizations that we prove in this paper in the rest of this section.

\paragraph{Treedepth and matched treedepth} The matched treedepth of a graph is closely related to its treedepth. The constant space algorithm based on treedepth is essentially an divide-and-conquer algorithm over a elimination tree of the pattern graph that executes a brute-force search over each root-to-leaf path in the elimination tree. Therefore, it runs in time $O(n^d)$. We exploit the fact that the elimination tree is matched, which forces an additional constraint that the vertices in each root-to-leaf path has to be covered by a matching. This allows the brute-force part of the algorithm to discover all $d$ vertices on the path using only $d/2$ edges. The central algorithm that we use to obtain constant space algorithms is given below:
\begin{restatable}{theorem}{counthommtd}
    \label{thm:count-hom-mtd}%
    Let $G$ be a graph with $\mtd(G) = d$, then given an $m$-edge graph $H$ as input, we can count the number of homomorphisms from $G$ to $H$ in $\otilde(m^{\lceil d/2 \rceil})$-time and constant space.
\end{restatable}
It is well-known that the number of $G$-subgraphs, for any $G$, can be expressed as a linear combination of the number of homomorphisms from a related set of graphs called the \emph{spasm of $G$}. The spasm of $G$ contains exactly all graphs that can be obtained by iteratively merging the independent sets in $G$. (See \ref{def:spasm} for formal definition). Although the treedepth of the spasm of $C_{11}$ is bounded by $6$, however, the matched treedepth is not necessarily bounded by the treedepth. We analyze all graphs in the spasm of $C_{10}$ and $C_{11}$ (there are $501$ such graphs) and show that the matched treedepth of each graph is at most $6$. This yields the following algorithm:
\begin{restatable}{theorem}{ckcountmtd}
    \label{thm:ckcountmtd}%
    Given an $m$-edge graph $H$ as input, we can count the number of $C_{k}$, where $k \leq 11$, as subgraphs in $\otilde(m^3)$-time and constant space.
\end{restatable}
For comparison, the brute-force algorithm takes $O(m^6)$-time and constant space; and the treedepth based algorithm takes $O(n^6)$-time and constant space.

As seen from the proof of our algorithm for counting $C_{11}$, the spasm of a pattern can contain a large number of graphs even for relatively small patterns. Therefore, it would be nice to have theorems that upper-bound the matched treedepth. Unfortunately, the property $\mtd(G) \leq k$ is not even subgraph-closed unlike treedepth. For example, it can be proved that $\mtd(K_4 - e) = 3$ but $\mtd(C_4) = 4$. However, interesting structural observations can still be made for matched treedepth. The following is a theorem that upper-bounds matched treedepth in terms of treedepth.
\begin{restatable}{theorem}{mtdtd}
    \label{thm:mtdtd}%
    For any graph $G$, $\mtd(G) \leq 2\cdot\td(G) - 2$.
\end{restatable}

Theorem~\ref{thm:mtdtd} implies that our constant-space algorithms from Theorem~\ref{thm:count-hom-mtd} for counting homomorphisms are asymptotically faster for \emph{all} patterns, where the inputs are sparse host graphs, when compared to the treedepth-based algorithm.

The following theorem shows that the time complexity for counting homomorphisms of a pattern is lower-bounded by the time complexity for counting all of its induced subgraphs.

\begin{restatable}{theorem}{mtdsub}
    \label{thm:mtdsub}
    Let $G$ be a graph and $G'$ is a connected, induced subgraph of $G$, then:
    \begin{enumerate}
        \item $\mtd(G') \leq \mtd(G)$ if $\mtd(G)$ is even.
        \item $\mtd(G') \leq \mtd(G) + 1$ if $\mtd(G)$ is odd.
    \end{enumerate}
\end{restatable}

In light of the importance of matched treedepth, it becomes crucial that we understand this structural parameter as much as possible. The graphs of treedepth $2$ are exactly the class of star graphs. This is also the class of graphs with matched treedepth $2$. However, for the graph $C_4$, we have $\td(C_4) = 3$ and $\mtd(C_4) = 4$. So it is interesting to know what are exactly the graphs where treedepth and matched treedepth coincide. The following theorem should be viewed as giving us a preliminary understanding of the relationship between these two parameters.
\begin{restatable}{theorem}{tdmtdthree}
    \label{thm:tdmtdthree}%
    Let $G$ be a graph such that $\td(G) = 3$. Then $\mtd(G) = 3$ if and only if $G$ is $(C_4, P_6, \graphw)$-free.
\end{restatable}
The graph $\graphw$ is the $(3, 3)$ tadpole graph (See Figure \ref{fig:mtd3forbidden}).

\paragraph{Treewidth and matched treewidth} The treewidth-based dynamic programming algorithm of D\'{\i}az, Serna, and Thilikos \cite{10.5555/646719.702252} can be strengthened to output an \emph{arithmetic circuit} that computes the \emph{homomorphism polynomial} for the pattern. An arithmetic circuit is a directed acyclic graph where each internal node is labeled $+$ or $\times$, each leaf is labeled by a variable or a field constant, and there is a designated output node. Such a graph computes a polynomial over the underlying field in a natural fashion. We find that by using a dynamic programming algorithm over matched tree decompositions, we can improve the size of the arithmetic circuit for \emph{sparse} host graphs. Our central theorem is given below:
\begin{restatable}{theorem}{counthommtw}
    \label{thm:counthommtw}%
    Let $G$ be a graph with $\mtw(G) = t$, then given an $m$-edge host graph $H$ as input, we can construct an arithmetic circuit computing the homomorphism polynomial from $G$ to $H$ in time $\otilde(m^{\lceil (t+1)/2 \rceil})$.
\end{restatable}
For graphs where matched treewidth and treewidth coincide, the running time for counting homomorphisms is a quadratic improvement on the algorithm by D\'{\i}az, Serna, and Thilikos \cite{10.5555/646719.702252} for sparse graphs. Therefore, this is also the best possible improvement one can hope to get without improving upon the algorithm by D\'{\i}az, Serna and Thilikos \cite{10.5555/646719.702252}. What is the worst case? The following theorem implies that the resulting algorithm \emph{cannot} be worse on sparse host graphs.
\begin{restatable}{theorem}{mtwvstw}
    \label{thm:mtwvstw}%
    For any graph $G$, we have $\mtw(G) \leq 2 \cdot \tw(G) + 1$.
\end{restatable}

Unfortunately, unlike for treewidth, the parameter $\mtw(G)$ is not monotone over the subgraph partial order. We first observe an explicit graph family with lower $\tw$ and larger $\mtw$. Consider the complete bipartite graph $K_{n,n}$ on $n$ vertices. Notice that $tw(K_{n, n}) = n$.
\begin{restatable}{proposition}{mtwknn}
	 $\mtw(K_{n,n})$ = $2n-2$ for all $n>1$.
\end{restatable}

The following observation shows that there exists supergraphs of $K_{n,n}$ with lower $\mtw$ than that of $K_{n, n}$.
\begin{restatable}{observation}{mtw-is-not-subgraph-closed}
Consider the supergraph $G$ of $K_{n,n}$ such that $V(G)=V(H)$, and there are edges in one partition of $K_{n, n}$ such that the independent set of size $n$ becomes a path on $n$ vertices. Note that although $\mtw(K_{n,n}) = 2n-2$, but $\mtw(G) = n$.
\end{restatable}

We show how we can use structural theorems about matched treewidth to prove algorithmic upper bounds. For example, to count the number of $P_{10}$ subgraphs, we only have to show that all graphs in the spasm of $P_{10}$ have low matched treewidth. The spasm of $P_{10}$ is a large set that contains more than $300$ graphs. Indeed, it is possible to analyze the matched treewidth for each of these graphs individually. However, it would be better if we have theorems that eliminate such tedious work.

We derive some structural theorems for low values of matched treewidth. Graphs with matched treewidth $1$ are exactly trees. We also show $\tw(C_5) = 2$ and $\mtw(C_5) = 3$.

We characterize the matched treewidth of partial $2$-trees using forbidden induced minors (See Definition~\ref{def:minor-ind-minor})  wherever possible. We show that $C_5$ is exactly the obstruction that forces higher matched treewidth for partial 2-trees.
\begin{restatable}{theorem}{twmtwtwo}
    \label{thm:twmtwtwo}%
    For any partial 2-tree $G$, the graph $G$ is $C_5$-induced-minor-free if and only if $\mtw(G) = 2$.
\end{restatable}
Notice that $\tw(G) = 2$ yields $O(n^3)$-time algorithms for counting homomorphisms. Even if $\mtw(G) = 3$, we obtain $\otilde(m^2)$-time algorithms for counting homomorphisms which is an improvement for sparse graphs. Does all treewidth $2$ graphs have matched treewidth at most $3$? No. The graph $X$ in Figure~\ref{fig:x} has treewidth $2$ and matched treewidth $4$ (See Observation~\ref{xmtw}). In fact, we can prove that $X$ is exactly the obstruction that forces treewidth $2$ graphs to have matched treewidth $4$.
\begin{figure}[ht]
\centering
 \begin{tikzpicture}=[circle,draw,thick,fill=white,inner sep=0.5mm]

	 \node[vertex, label=$u_0$] (u1) at (60+30:2) {};
	 \node[vertex, label=$u_1$] (u2) at (60*2+30:2) {}; 
	 \node[vertex, label=left:$u_2$] (u3) at (60*3+30:2) {}; 
	 \node[vertex, label=below:$u_3$] (u4) at (60*4+30:2) {}; 
	 \node[vertex, label=right:$u_4$] (u5) at (60*5+30:2) {}; 
	 \node[vertex, label=$u_5$] (u6) at (60*6+30:2) {}; 
	 
	 \node[vertex, label=$u'_1$] (v1) at (120+30:1) {}; 
	 \node[vertex, label=$u'_3$] (v2) at (240+30:1) {}; 
	 \node[vertex, label=$u'_5$] (v3) at (30:1) {};

	\draw (u1)--(u2)--(u3)--(u4)--(u5)--(u6)--(u1);
	\draw  (u1)--(v1)--(u3)--(v2)--(u5)--(v3)--(u1);
\end{tikzpicture}
\caption{\centering The graph $X$.}\label{fig:x}
\end{figure}

\begin{restatable}{theorem}{xminorfree}\label{thm:mtwHMF}
    \label{thm:xminorfree}%
    For any partial 2-tree $G$, the graph $G$ is $X$-induced-minor-free if and only if $\mtw(G) \leq 3$.
\end{restatable}
This theorem implies that all $X$-induced-minor-free, treewidth $2$ patterns have $\otilde(m^2)$-time homomorphism counting algorithms. This is an improvement for sparse host graph even over the fast matrix multiplication based algorithm given by Curticapean, Dell, and Marx \cite{10.1007/978-3-642-32241-9_4} for counting homomorphisms from treewidth $2$ graphs that runs in $O(n^\omega)$-time. Since the spasm of $P_{10}$ does not contain any treewidth $4$ graph or graph with an $X$-induced minor, we can show that there is an $\otilde(m^2)$-time algorithm for counting subgraph isormophisms of all paths on at most $10$ vertices by showing that all treewidth $3$ graphs in the spasm of $P_{10}$ has matched treewidth $3$. There are only $18$ such graphs. Analyzing their matched treewidth yields the following theorem:
\begin{restatable}{theorem}{pkcountmtw}
    \label{thm:pkcountmtw}%
    Given an $m$-edge graph $H$ as input, we can count the number of $P_k$ subgraphs, where $k \leq 10$, in $\otilde(m^2)$-time.
\end{restatable}
To the best of our knowledge, the best known path counting algorithms take $\Omega(n^4)$ time for paths on $10$ vertices. Therefore, our algorithm is a significant improvement for sparse host graphs and no worse than the best known algorithm for dense host graphs. An easy corollary of the proof of this result is given below:
\begin{restatable}{theorem}{ckcountmtw}
    \label{thm:ckcountmtw}%
    Given an $m$-edge graph $H$ as input, we can count the number of cycles of length at most $9$ in $\otilde(m^2)$-time.
\end{restatable}
These cycle counting algorithms are an improvement on sparse graphs over the $O(n^\omega)$-time algorithms for cycles of length at most $8$ given by Alon, Yuster and Zwick \cite{Alon1997}.

We also show how to use our improved homomorphism polynomial construction algorithm to speed up detection of induced subgraphs. In particular, we show the following:
\begin{restatable}{theorem}{csixind}
    \label{thm:csixind}%
    Given an $m$-edge host graph as input, we can find an induced $C_6$ or report that none exists in $\otilde(m^2)$-time.
\end{restatable}
This algorithm is no worse than the $O(n^4)$ time algorithm that can be derived using the techniques by Bl\"aser, Komarath, and Sreenivasaiah \cite{DBLP:conf/fsttcs/BlaserKS18}. For sparse graphs, our algorithm provides a quadratic improvement. We also show the following:
\begin{restatable}{theorem}{pkbarind}
    \label{thm:pkbarind}%
    Given an $m$-edge host graph as input, we can find an induced $\overline{P_k}$ or report that none exists in $\otilde(m^{(k-2)/2})$-time.
\end{restatable}
This is also a quadratic improvement over the $O(n^{k-2})$ time algorithm given by Bl\"aser, Komarath, and Sreenivasaiah \cite{DBLP:conf/fsttcs/BlaserKS18} when the host graph is sparse. These algorithms are obtained by analyzing the matched treewidth \emph{and} the automorphism structure of a set of graphs derived from the pattern.

From a technical standpoint, we see our algorithms as a natural combination of pattern detection and counting algorithms that work well on sparse host graph such as the $\otilde(m)$ algorithm for counting $k$-walks, the $\otilde(m^{k/2})$ algorithm for counting $k$-cliques, and $\otilde(m^{(k-1)/2})$ time algorithm for detecting induced $K_k-e$ by Vassilevska \cite{Virginia08} and insights that improve the running-time on dense graphs by exploiting structural parameters treedepth and treewidth. We do not make use of fast matrix multiplication in any of our algorithms. Such algorithms, called \emph{combinatorial algorithms}, are also of general interest to the community.

\subsection{Related work}
Algorithms for counting \emph{induced subgraphs} are related to the problems that we consider but we do not consider any algorithms for it. This problem seems to be much harder. It is conjectured by Floderus, Kowaluk, Lingas, Lundell \cite{10.1007/978-3-642-32241-9_4} that counting induced subgraphs for any $k$-vertex pattern graph is as hard as counting $k$-cliques for sufficiently large $k$. Several works have considered the parameterized complexity of counting subgraphs (See \cite{DBLP:conf/focs/CurticapeanM14, 10.1145/3055399.3055502, roth_et_al:LIPIcs.ICALP.2021.108, DBLP:conf/stoc/FockeR22, DBLP:journals/algorithmica/DorflerRSW22}) where the primary goal is to obtain a dichotomy of easy vs hard based on structural graph parameters. Some works have also considered restrictions on host graphs such as $d$-degeneracy \cite{BR21}. The papers on parameterized complexity primarily chooses to focus on the growth-rate of the exponent for a family of patterns such as $k$-paths, $k$-cycles, or $k$-cliques and not the exact exponent for small graphs as we do in this paper.

\section{Preliminaries}
\label{sec:prelims}

We consider simple graphs. We refer the reader to Douglas West's textbook \cite{west_introduction_2000} for basic definitions in graph theory. We use the following common notations for some well-known graphs: $P_k$ for $k$-vertex paths, $C_k$ for $k$-cycles, $K_k$ for $k$-cliques, $K_k - e$ for $k$-clique with one edge missing, $K_{m,n}$ for complete bipartite graphs. A $k$-star is a $(k+1)$-vertex graph with a vertex $u$ adjacent to vertices $v_1, \dotsc, v_k$ and no other edges. A \emph{star graph} is a $k$-star for some $k$. For a graph $G$ and $S \subseteq V(G)$, we denote by $G[S]$ the subgraph of $G$ induced by the vertices in $S$.

\begin{definition}
\label{def:hom}
    Given two graphs $G$ and $H$, a graph homomorphism from $G$ to $H$ is a map $\phi: V(G) \rightarrow V(H)$ such that if $uv \in E(G)$, then $\phi(u)\phi(v) \in E(H)$.
\end{definition}

We denote by Hom($G$, $H$), the set of all homomorphisms from $G$ to $H$.

\begin{definition}\label{def:hompoly} 
Given two graphs $G$ and $H$, a homomorphism polynomial is an associated polynomial $\ghom_G[H]$ such that there is a one-to-one correspondence between its monomials and the homomorphisms from $G$ to $H$. We define:
\begin{equation*}
    \ghom_G[H] = \sum_{\phi \in Hom(G, H)} \prod_{u \in V(G)}y_{\phi(u)} \prod_{\{u, v\} \in E(G)} x_{\{\phi(u), \phi(v)\}}
\end{equation*}

Note that $H$ has a subgraph isomorphic to $G$ if and only if $P_G$ has a multilinear monomial.

\end{definition}

We say that a graph $G'$ is $(G_1, \dotsc, G_m)$-free if no induced subgraph of $G'$ is isomorphic to $G_i$ for any $i$. We denote the complement of a graph $G$ by $\overline{G}$. We have $V(G) = V(\overline{G})$ and the edges of $\overline{G}$ are exactly the non-edges of $G$ and vice versa.

We assume that all pattern graphs are connected. Since our primary algorithms are all based on counting homomorphisms, this does not lose generality as the number of homomorphisms from a disconnected pattern is just the product of the number of homomorphisms from its components.

\begin{definition}\label{def:ET} An \emph{elimination tree} $(T, r)$ of a connected graph $G$ is a tree rooted at $r \in V(G)$, where the sub-trees of $r$ are recursively elimination trees of the connected components of the graph $G \setminus r$. The elimination tree of an empty (no vertices or edges) graph is the empty tree. The depth of an elimination tree $(T,r)$ is defined as the maximum number of vertices over all root-to-leaf paths in $T$.  The \emph{treedepth} of a graph $G$, denoted $\td(G)$, is the minimum depth among all possible elimination trees of $G$. 
\end{definition}

Intuitively, treedepth measures the closeness of a given graph to star graphs which are exactly the connected graphs having treedepth $2$. We introduce a related notion called matched treedepth that seems to be helpful when designing algorithms for finding or counting patterns in sparse host graphs.
\begin{definition}\label{def:matchedET}
A \emph{matched elimination tree} for a graph $G$ is an elimination tree such that the following conditions are true for any root-to-leaf path $(v_1, \dotsc, v_k)$:

\begin{itemize}
    \item If $k$ is even, then $v_1v_2, v_3v_4, \dotsc, v_{k-1}v_k$ is a matching in $G$.
    
    \item If $k$ is odd, then there is some $i$ such that $E' = \{v_1v_2, \dotsc, v_{i-1}v_i, v_iv_{i+1}, \dotsc, v_{k-1}v_k\}$ and $E' \subseteq E(G)$. We have that $E' \setminus \{v_{i-1}v_i, v_iv_{i+1}\}$ is a matching on $(k-3)/2$ vertices.
\end{itemize}

The \emph{matched treedepth} of a graph $G$, denoted $\mtd(G)$, is the minimum depth among all possible matched elimination trees of $G$.
\end{definition}

The matched treedepth is always finite (See Proposition~\ref{prop:mtdmm}).

\begin{definition}
\label{def:tree-decomposition-treewidth}%
A \emph{tree decomposition} of a graph $G$ is a pair $(T, B(t)_{t \in T})$ where $T$ is a tree and $B(t)$, called a \emph{bag}, is a collection of subset of vertices of $G$ corresponding to every node $t \in T$. 

\begin{itemize}
    \item (Connectivity Property) For all $v \in V(G)$, there is a node $t \in V(T)$ such that $v \in B(t)$ and all such nodes $t$ that contain $v$ form a connected component in $T$.
    \item (Edge Property) For all $e \in E(G)$, there is a node $t \in V(T)$ such that $e \subseteq B(t)$.
\end{itemize}
The \emph{width of a tree decompostion} $(T, B)$ is defined as the maximum bag size minus one, that is, $\max_{t \in T}|B(t) - 1|$. The \emph{treewidth} of a graph $G$, $tw(G)$, is the minimum possible width among all possible tree decompositions of $G$. 
\end{definition}
Intuitively, treewidth measures the closeness of the given graph to trees which are exactly the graphs with treewidth $1$. We introduce a related notion, called matched treewidth, closely related to treewidth, that seems to be useful for designing dynamic programming algorithms over sparse host graphs.
\begin{definition}\label{def:matched-tree-decomposition-treewidth}%
A \emph{matched tree decomposition} for a graph $G$ is a tree decomposition where for every bag in the tree decomposition, the subgraph of $G$ induced by the vertices in that bag has either a perfect matching or a matching where exactly one vertex $v$ in the bag is unmatched and $v$ is adjacent to some vertex in the matching. We call such bags \emph{matched}. The \emph{matched treewidth} of a graph $G$, $mtw(G)$, is the minimum possible width among all possible matched tree decompositions of $G$.
\end{definition}

Matched treewidth is finite for all graphs (See Proposition~\ref{prop:mtwmtd}). This is not trivial unlike treewidth because a single bag tree decomposition that contains all the vertices in the graph need not be matched.

We call a tree decomposition \emph{reduced} if no bag $B$ is a subset of another bag. Given any tree decomposition $T$, we can obtain a reduced tree decomposition $T'$ such that the width of $T'$ is at most the width of $T$. Moreover, all bags in $T'$ are also bags in $T$. This implies that if $T$ is matched, then $T'$ is matched as well.

There are several equivalent characterizations for treewidth. Below, we state the ones that we use in this paper.

\begin{definition}
  A \emph{$k$-tree} is a graph formed by starting with a $(k+1)$-clique and repeatedly adding a vertex connected to exactly $k$ vertices of the existing $(k+1)$-clique. A \emph{partial} of a graph $G$ is a graph obtained by deleting edges from $G$. The set of all graphs with treewidth at most $k$ is exactly the class of partial $k$-trees.
\end{definition}

We can construct a \emph{standard tree decomposition} $T$ for any $k$-tree as follows: The root bag of $T$ contains the vertices in the initial $(k+1)$-clique. Let $S$ be a $k$-sized subset of this clique such that a new vertex $v$ is added to the $k$-tree by connecting it to all vertices in $S$. Then, we add a sub-tree to $T$ that will be a standard tree decomposition of the $k$-tree constructed using $S\cup \{v\}$ as the starting $(k+1)$-clique.

A \emph{chordal completion} of a graph $G$ is a super-graph $G'$ of $G$ such that $G'$ has no induced cycles of length more than $3$. A chordal completion that minimizes the size of the largest clique is called \emph{minimum chordal completion}. The treewidth of a graph $G$ is the size of the largest clique in its minimum chordal completion.
  
Two paths $P_1$ and $P_2$ from $u$ to $v$ are \emph{internally disjoint} if and only if $P_1$ and $P_2$ do not have any common internal vertex.

A graph $G$ is $\emph{2-connected}$ or \emph{biconnected}, if for any $x\in V(G)$, $G - x$ is connected. Equivalently, for any two vertices in $G$, there are at least $2$ internally disjoint paths in $G$. 
\begin{definition} \label{def:sp-graph}
 A series-parallel graph is a triple $(G, s, t)$ where $s$ and $t$ are vertices in $G$. This class is recursively defined as follows:
  \begin{itemize}
      \item An edge $\{s, t\}$ is a series-parallel graph.
      \item (series composition) If $(G_1, s_1, t_1)$ and $(G_2, s_2, t_2)$ are series-parallel graphs, then the graph obtained by identifiying $s_2$ with $t_1$ is also series-parallel.
      \item (parallel composition) If $(G_1, s_1, t_1)$ and $(G_2, s_2, t_2)$ are series-parallel graphs, then the graph obtained by identifiying $s_1$ with $s_2$ and $t_1$ with $t_2$ is also series-parallel.
  \end{itemize}
  A graph has treewidth at most $2$ is if and only if all of its biconnected components are series-parallel graphs.
\end{definition}  

\begin{definition}\label{def:minor-ind-minor}
A graph $G$ is said to be a \emph{minor} of a graph $G^\prime$ if $G$ can be obtained from $G^\prime$ either by deleting edges/vertices,  or by contracting the edges. (The operation of contraction merges two adjacent vertices $u$ and $v$ in the graph and removes the edge $(u,v)$.) If $G$ is obtained from an induced subgraph of $G^\prime$ by contracting the edges, then it is said to be an \emph{induced minor} of $G^\prime$.
\end{definition}

A graph $G$ is called $G'$-induced-minor-free ($G'$-minor-free) if $G'$ is not an induced minor (resp. minor) of $G$.

An edge subdivision is an operation which deletes the edge $(u, v)$  and adds a new vertex $w$ and the edges $(u, w)$ and $(w, v)$. A graph $G^\prime$ obtained from $G$ by a sequence of edge subdivisions is said to be a \emph{subdivision} of $G$.

\subsection{Representation of graphs}
We assume the following time complexities for basic graph operations. Any representation that satisfies these is sufficient.
\begin{itemize}
    \item Given $u$ and $v$, it can be checked in $\otilde(1)$-time whether $uv$ is an edge.
    \item Iterating over all the edges $xy \in E(H)$ ordered by $x$ can be done in $\otilde(m)$-time, where $m$ is the number of edges in $H$.
\end{itemize}
These requirements are satisfied by the following adjacency-list representation. To represent a graph $H$, we store a red-black tree $T$ that contains non-isolated vertices of $H$ where vertices are ordered according to their labels. Consider a node in this tree that corresponds to a vertex $u$. This node stores another red-black tree $T_u$ that stores all neighbors of $u$ in $H$. Now, to check whether $uv$ is an edge or not, we perform a lookup for $u$ in $T$ followed by a lookup for $v$ in $T_u$ if $u$ was found. We can iterate over all edges $xy$ ordered by $x$ by performing an inorder traversal of $T$ where for each node $u$, we perform an inorder traversal of $T_u$.

If the pattern does not contain any isolated vertices, then we can ignore isolated vertices in the host graph as well. If the pattern is $G = G' + v$, where $v$ is an isolated vertex and $G'$ is any graph, then the number of homomorphisms from $G$ to $H$ is obtained by multiplying the number of homomorphisms from $G'$ to $H$ by $n$, where $n$ is the number of vertices in $H$. This can be calculated by simply storing the number of vertices of $H$ in the data structure.

\section{Matched treedepth}
\label{sec:mtd}

In this section, we introduce algorithms that count homomorphisms and subgraphs efficiently in constant space on sparse host graphs. The central theorem in this section is given below.

\counthommtd*

\begin{proof}
The algorithm is given in Algorithm~\ref{alg:count-hom-mtd}. We can compute the result needed by calling $\textsc{COUNT-HOM-MTD}(G, E, r, H, \phi)$, where $E$ is an elimination tree for $G$ of depth $d$, $r$ is the root vertex in $E$, and $\phi$ is the empty homomorphism. For simplicity of presentation, we assume that each root-to-leaf path in $E$ has an even number of vertices. Odd number of vertices in a root-to-leaf path is handled similarly.

We assume that the host graph $H$ is represented using a symmetric adjacency list representation. This is mainly to ensure that we can iterate over all edges $xy$ in $H$ ordered by $x$ in Line~\ref{line:mtdmainloop}.

First, we prove that the algorithm is correct. We claim that the call $\textsc{COUNT-HOM-MTD}$ $(G, E, v, H, \sigma)$ where the parameters are as specified in the algorithm returns the number of homomorphisms from $G_v$ to $H$ that extends $\sigma$. This is proved by an induction on the height of $v$ in $E$. Since $v$ is a top node, the base case is when the height is $2$. In this case, $G_v$ is a star graph and it is easy to see that the algorithm works. We now prove the inductive case. The variable $t$ computes the final answer. Denote by $s_{u, x}$ for vertex $x$ in $H$ the number of homomorphisms from $G_u$ to $H$ that extends $\tau = \sigma \cup \{v \mapsto x\}$. Notice that since $uv \in E(G)$, to extend $\tau$, the vertex $u$ must be mapped to some $y$ such that $xy \in E(H)$. Therefore, iterating over all such $y$ is sufficient. Notice that we can compute $t$ as $\sum_{\tau} \prod_u s_{u, x}$. However, this would need storing $|V(G)||V(H)|$ variables. By iterating over the edges of $H$ ordered by $x$, we can afford to reuse a single $s_u$ for different $x$ instead of keeping a separate $s_{u, x}$ for each $x$. Now, we show that $s_u$ correctly computes $s_{u, x}$ once the main loop finishes with an $x$. By the inductive hypothesis, the variable $c_w$ is the number of homomorphisms from $G_w$ to $H$ that extends $\sigma' = \sigma \cup \{ v \mapsto x, u \mapsto y\}$. Therefore, we have $s_{u, x} = \sum_y \prod_w c_w$. Line~\ref{line:sucompute} correctly computes this into $s_u$. Line~\ref{line:tcompute} correctly updates $t$ once an $x$ is finished. Finally, we reset $s_u$ to $0$ before processing the next $x$.

Now, we prove the running-time and space usage of the algorithm. Notice that the depth of the recursion is bounded by the depth of the elimination tree $E$ and each level of recursion stores only constantly many variables. Therefore, the space usage is constant. The main loop in Line~\ref{line:mtdmainloop} runs for $2m$ iterations. The inner-loops only have a constant number of iterations. Therefore, the recursive calls are made only $O(m)$ times. We process two levels of the elimination tree in a level. Therefore, the total running-time is given by $t(d) \leq O(m)t(d-2) + \otilde(1) = \otilde(m^{d/2})$.

\newpage
\begin{algorithm}
\caption{COUNT-HOM-MTD(G, E, v, H, $\sigma$)\label{alg:count-hom-mtd}}
\begin{algorithmic}[lines=1]
\Require{$G$ - The pattern graph}
\Require{$E$ - Matched elimination tree for $G$}
\Require{$v$ - A \emph{top} vertex in $E$}
\Require{$H$ - The host graph}
\Require{$\sigma$ - A partial homomorphism from the ancestors of $v$ to $H$}
\State $t\gets 0$
\State $s_u \gets 0$ for all children $u$ of $v$
\ForAll{edges $xy \in E(H)$ ordered by $x$}\label{line:mtdmainloop}
\ForAll{children $u$ of $v$ in $E$}
  \State $\sigma' \gets \sigma \cup \{v \mapsto x, u \mapsto y\}$
  \If{$\sigma'$ is an invalid homomorphism}
    \State \textbf{continue}
  \EndIf\label{SigmaPrimeCheck}
  \ForAll{children $w$ of $u$}
    \State $c_w \gets \text{COUNT-HOM-MTD}(G, E, w, H, \sigma')$
  \EndFor
  \State $s_u \gets s_u + \prod_w c_w$ \label{line:sucompute}
\EndFor\label{ChildLoop}
\If{$xy$ is the last edge on $x$}
  \State $t \gets t + \prod_u s_u$ \label{line:tcompute}
  \State $s_u \gets 0$ for all $u$ \label{line:sureset}
\EndIf\label{LastX}
\EndFor\label{MainLoop}
\State \Return $t$
\end{algorithmic}
\end{algorithm}
\end{proof}

A fundamental question regarding matched treedepth is whether it is finite for all graphs. It is, as the following proposition shows.
\begin{proposition}
    \label{prop:mtdmm}
    For any graph $G$, we have $\mtd(G)$ is at most 1 + the number of vertices in the smallest maximal matching in $G$.
\end{proposition}
\begin{proof}
    We partition $V(G)$ into a maximal matching $M = (v_1w_1, \dotsc, v_mw_m)$ and an independent set $\{u_1, \dotsc, u_k\}$. A matched elimination tree $T$ for $G$ can be constructed as follows: Put the vertices $v_1, w_1, \dotsc, v_m, w_m$ on a root-to-leaf path in that order. Let us call this path the spine. For each vertex $u_i$, make the lowest vertex in the spine adjacent to $u_i$ in $G$ its parent in $T$. It is easy to see that $T$ is a matched elimination tree and its depth is at most 1 + the number of vertices in $M$.
\end{proof}
Although, Proposition~\ref{prop:mtdmm} proves an upper-bound for matched treedepth. It is not very useful from an algorithmic perspective as it is easy to see that there is an $O(m^{k+1})$-time, constant space algorithm for counting patterns with maximal matchings of $k$ edges. Any smallest maximal matching in $C_{11}$, for example, has $4$ edges. However, the proof of Theorem~\ref{thm:ckcountmtd} given below shows that we can do better. The algorithm is based on the well-known technique of expressing subgraph count as a linear combination of homomorphism counts.

\begin{definition}
\label{def:spasm}
Let $\mathcal{I}$ be the set of all the independent sets in the graph $G$. For some $I \in \mathcal{I}$$, \merge(G,I)$ is the graph obtained by merging the vertices of $I$. Then
\begin{equation*}
\spasm(G) = \{G\} \cup \bigcup_{I \in \mathcal{I}} \spasm(\merge(G, I))
\end{equation*}
\end{definition}
For any pattern graph $G$, it turns out that the number of subgraph isomorphisms from $G$ to a host graph $H$ are just the linear combination of all possible graph homomorphisms from $\spasm{G}$ to $H$. That is, there exists constants $\alpha_{G'} \in \mathbb{Q}$ such that:
\begin{equation}\label{eq:subcount}
\gsub_G[H] = \sum_{G'} \alpha_{G'} \ghom_{G'}[H]
\end{equation}
where $G'$ ranges over all graphs in $\spasm(G)$. This equation is used to count subgraphs by many authors (See for example \cite{Alon1997, 10.1145/3055399.3055502}).

\ckcountmtd*
\begin{proof}
    We analyzed all graphs in $\spasm(C_{11})$ and $\spasm(C_{10})$ and concluded that each of them has matched treedepth at most $6$. A pdf that contains all these graphs and their corresponding matched elimination trees can be found at (\url{https://github.com/anonymous1203/Spasm}). For seeing that the stated algorithms exist for smaller cycles, observe that $\spasm(C_k) \subset \spasm(C_{k+2})$ for $k \geq 3$.
\end{proof}

The class of treedepth $2$ graphs is exactly the star graphs. These graphs have an $O(n^2)$-time, constant space, treedepth-based algorithm. This class of graphs is also the graphs of matched treedepth $2$. It is easy to see that the number of homomorphisms from star graphs can also be counted in $O(m)$-time and constant space using the following observation. Let $d_1, \dotsc, d_n$ be the degrees of vertices in the graph. Then, the count is given by the expression $\sum_i d_i^k$. Note that this is an asymptotic improvement for sparse graphs. We now show that this asymptotic improvements exists for all patterns. First, we need a restricted form of elimination trees.
\begin{definition}
    An elimination tree $T$ is \emph{connected} if for every node $u$ in $T$ and a child $v$ of $u$ in $T$, $u$ is adjacent to some node in $T_v$
\end{definition}
Now, we show that connected elimination trees are optimal.
\begin{lemma}\label{lem:connet}
    Every connected graph $G$ has a connected elimination tree of depth $\td(G)$.
\end{lemma}
\begin{proof}
    We show how to construct a connected elimination tree $T'$ from an elimination tree $T$ without increasing its depth. Let $T' = T$ initially. Suppose there exists some node $u$ in $T'$ that violates the property (If not, we are done). Then, there exists a child $v$ of $u$ in $T'$ such that there is no edge in $G$ between $u$ and any node in $T'_v$. Let $w$ be a node in $T'_v$ such that $w$ is adjacent to some proper ancestor $x$ of $u$. Such a $w$ must exist because $G$ is connected. Now, in $T'$, remove the subtree $T'_v$ and make it a subtree of the node $x$. Repeat this process until all nodes in $T'$ satisfy the required property. This process must terminate since at each step, we reduce the number of nodes violating the property by at least one. This process cannot increase the depth of $T'$ because the only modification is to move a subtree upwards to be a subtree of a proper ancestor of its parent.
\end{proof}

\mtdtd*
\begin{proof}
    We start with a connected elimination tree $T$ with depth $d$ of $G$ and show how to construct a matched elimination tree of $G$ from $T$. We use induction on $d$. For $d=2$, the tree is already matched and has depth $2 = 2\cdot 2 - 2$.
    
    Our construction is iterative and top-down. Each iteration ensures that the current top-most node in the elimination tree is adjacent, in the graph $G$, to all its children and that the elimination tree is connected.
    
    Each iteration consists of two phases iteratively executed until the desired property is satisfied. The first phase ensures that the root node $r$ is adjacent in $G$ to all its children in $T$. If $r$ has a child $v$ that is not adjacent in $G$ to $r$, then since $T$ is connected, there is some node $w$ in $T_v$ such that $rw \in E(G)$. Then, we make $w$ a child of $r$ in $T$, delete $w$ from $T_v$, and make all children of $w$ children of parent of $w$. The resulting tree is an elimination tree of depth at most $d+1$. After this phase, the root node is adjacent in $G$ to all its children, the tree's depth has increased by at most one. However, it may not be connected.
    
    In the second phase, we use the construction of Lemma~\ref{lem:connet} to make the tree connected without increasing the depth. Observe that the construction will keep the existing children of root as is and may add new chlidren to $r$ that are not adjacent to $r$ in $G$. Suppose a new node $u$ was added as a new child to $r$ in this second phase. The height of subtree rooted at $u$ is at most $d-1$. Therefore, the tree $(r, T_u)$ obtained by attaching $u$ to $r$ has height at most $d$. We can now execute phase 1 on all the trees $(r, T_u)$ for all such $u$ without increasing depth beyond $d+1$. This process must eventually terminate as we add at least one new child to the root every time.
    
    At the end of the iteration, consider a grandchild $x$ of $r$. If it is a leaf, since the tree is connected, $x$ must be adjacent in $G$ to its parent in $T$ and we are done. Otherwise, the subtree $T_x$ is a tree of depth at least $2$ and at most $d-1$ that is connected. So by the induction hypothesis, we obtain that $T_x$ is a matched elimination tree of depth at most $2(d-1) - 2$. This means that the original tree is converted to a matched elimination tree of depth at most $2 + 2(d-1) - 2 = 2d-2$ as required.
\end{proof}

\begin{corollary}
    Suppose $G$ has treedepth $d$. Given an $m$-edge host graph as input, we can count the number of homomorphisms from $G$ to the host graph in $\otilde(m^{d-1})$-time and constant space.
\end{corollary}
Notice that that above algorithm is asymptotically better than the treedepth based algorithm for all patterns on sparse host graphs.

We now prove that the matched treedepth of an induced subgraph cannot be much larger compared to that of the graph containing it. In other words, we can obtain lower bounds on the matched treedepth of a graph by obtaining lower bounds on the matched treedepth of any of its induced subgraphs. We prove the theorem for connected subgraphs. But note that the matched treedepth of a disconnected graph is the maximum of its connected components.

\mtdsub*
\begin{proof}
    We start with a matched elimination tree $T$ of even (The odd case is similar) depth $d$ for $G$ and construct a matched elimination tree for $G'$. First, we delete all nodes in the elimination tree that are in $G$ but not in $G'$. If a node $u$ in $T$ has parent $v$ that was deleted, then we make $u$ a child of the closest ancestor of $v$ that is still in $T$. The final forest thus obtained is a tree $T'$ because $G'$ is connected. We assume without loss of generality that $T'$ is a connected elimination tree.
    
    Suppose $r'$ is the root of $T'$. We will now modify $T'$ into a matched elimination tree. We will first analyze paths in $T'$ that correspond to even length root-to-leaf paths in $T$. If $T'$ is not matched on this path, then there exists a node $u$ closest to $r'$ such that $u$ is not connected in $G'$ to a child $v$ in $T'$. This is only possible if $u$ was either matched to a child $w$ of $u$ in $T$ or its parent $w$ in $T$ and $w$ is not in $G'$. Therefore, we have $\dist_{T'}(r', u) + \depth(T'_v) < \depth(T)$ and this means we can afford to increase the length of any path that passes through edge $uv$ by $1$. Since $T'$ is connected, we can now apply the transformation in the proof of Theorem~\ref{thm:mtdtd} to match $u$ with one of its descendants in $T'_v$. The depth is still at most $d$ because this transformation increases the depth by at most $1$ In effect, the increase in depth in this branch of the tree by pulling up a descendant is compensated by the fact that the unmatched vertex was introduced by deleting a vertex in this branch.
    
    We can iteratively apply the above construction to make $T'$ a matched elimination tree while keeping $T'$ connected. However, applying the above transformation may introduce a node $u$ in $T'$ that is not matched to a child $v$ \emph{because} $v$'s parent $w$ was pulled up in the tree to match with some other vertex. Such $u$ also satisfy the inequality $\dist_{T'}(r', u) + \depth(T'_v) < \depth(T)$. Why? Any path in the tree $T'$ that passed through both $u$ and $v$ earlier had to pass through $w$. But, the fact that $w$ was pulled up implies that these paths had length strictly less than $\depth(T)$ by the argument in the previous paragraph. And shifting $w$ to a position earlier in the path cannot increase its length. 
    
    For root-to-leaf paths in $T$ of odd length, there might be a matched $P_3$ on vertices $uvw$ such that $v$ is not in $G'$. In this case too, by the same argument, the transformations increase the depth to at most $d+1$ (In this case, we may have to pull up two distinct vertices for matching $u$ and $w$.). When $d$ is even, increasing the length of such paths by $1$ does not increase the depth of the tree. When $d$ is odd, increasing the length of such paths by $1$ increases the depth by atmost $1$.
\end{proof}

\begin{figure}[ht!]
    \centering
    \ctikzfig{c4p6x98b}
    \caption{Forbidden graphs for $\mtd(.) \leq 3$}
    \label{fig:mtd3forbidden}
\end{figure}
Theorem~\ref{thm:mtdtd} implies that all graphs $G$ with $\td(G) = 3$ also have $\mtd(G) \leq 4$. It is easy to see that $C_4$, $P_6$, and $\graphw$ given in Figure~\ref{fig:mtd3forbidden} have $\td(.) = 3$ and $\mtd(.) = 4$. By Theorem~\ref{thm:mtdsub}, it is also possible that some super graph of these graphs have $\mtd(.) = 3$. However, in the below theorem we show that this cannot happen and that the graphs $G$ with $\td(G) = \mtd(G) = 3$ are exactly the graphs where we forbid these graphs as induced subgraphs.

\tdmtdthree*
\begin{proof}
(Proof for the ``if'' direction) Let $G$ be a connected graph with $td(G) = 3$ such that $G$ is $(C_4, P_6, \graphw)$-free. We will construct a matched elimination tree for $G$ of depth $3$ from its elimination tree.

Consider an elimination tree $T$ of $G$ with depth $3$, rooted at some vertex $r$. Let $\{v_1, v_2, \ldots v_k\}$ be the children of $r$ in $T$. Notice that for any child $c$ of $v_i$ (for any $i$), if $c$ is not adjacent to $v_i$ in $G$ then it must be adjacent to $r$ in $G$ (else, $G$ is not connected). We make all such $c$, which are not adjacent to $v_i$ in $G$, a child of $r$ instead of $v_i$, in the elimination tree $T$. Note that this neither violates any property of elimination tree nor does it increase the depth of elimination tree. Also, after doing this modification, we can further assume that all the children of $v_i$ in $T$ are adjacent to $v_i$ in $G$.

Now, if $rv_i \in E(G)$ for all $i \in [k]$, then $T$ itself is a matched elimination tree and we are done. Suppose there exists an $i$ such that $rv_i \notin E(G)$. Since $G$ is connected, a path must exists from $r$ to $v_i$. Since $T$ has depth $3$, this path must be a $P_3$ and therefore $r$ and $v_i$ has at least one common neighbor. Any common neighbour $u$ of $r$ and $v_i$ must be a child of $v_i$ in $T$. Moreover, if $r$ and $v_i$ have two common neighbors, say $u_1$ and $u_2$, then $ru_1vu_2r$ is an induced $C_4$. So $r$ and $v_i$ must have exactly one common neighbor. We now consider the following cases: 

\begin{itemize}
    \item $r$ has exactly one child $v_1$.

    Let ${c_1, c_2, \dotsc, c_k}$ be the children of $v_1$. Recall that $rv_1\notin E(G)$, and they have exactly one common neighbor in $G$, say $c_1$. Then, we construct a matched elimination tree $T'$ as follows: The root of $T'$ is $c_1$, vertices $r$ and $v_1$ become the children of $c_1$, the children of $v_1$ in $T'$ are $\{c_2, \dotsc, c_k\}$ The tree $T^\prime$ is a valid matched elimination tree of depth $3$ for $G$.

    \item $r$ has multiple children, say $\{v_1, \ldots v_\ell\}$ for some $\ell > 1$. 
    
    We split this into two cases.
    
    \begin{itemize}
        \item ($r$ is not adjacent to exactly one of its children, say $v_1$)
        
        The vertex $v_1$ is not a leaf in $T$ since $G$ is connected. If $v_1$ has exactly one child, then, since it must be a common neighbor $c$ of $r$ and $v_1$, we swap $c$ and $v_1$, and this does not increase the depth of $T$. Else, suppose $v_1$ has at least two children,  namely $\{c_1, c_2, \ldots c_k\}$ (where the common neighbor with $r$ is $c_1$). We now argue that all the other children $\{v_2, \ldots v_l\}$ of $r$ must be leaves in $T$. If not, then say $v_2$ has a child $c_3$, then we encounter an induced $P_6$, namely $(c_2, v_1, c_1, r, v_2, c_3)$ (if $rc_3 \not\in E(G)$) or an induced $\graphw$ (if $rc_3 \in E(G)$), which is a contradiction. So, $\{v_2, \ldots v_l\}$ are leaf vertices. Now, we can convert $T$ to $T^\prime$ by rooting it at $c_1$ instead of $r$. The children of $c_1$ in $T^\prime$ are precisely $r$ and $v_1$, the children of $r$ are precisely all the leaves $\{v_2, \ldots v_l\}$, and the children of $v_1$ are all the same except $c_1$, that is, $\{c_2, \ldots c_k\}$. It is easy to see that $T^\prime$ is a matched elimination tree of depth $3$.
         
        \item ($r$ is not adjacent to at least two of its children, say $v_1$ and $v_2$ and maybe more)
 
        Due to the connectivity of $G$, the vertices $v_1$, $v_2$ cannot be leaves in $T$. If $v_1$ has exactly one child $c$, then $c$ must be a common neighbor of $r$ and $v_1$. We swap $v_1$ and $c$ in $T$. This will either fall into the previous case or we can assume $v_1$ has more than one child. Let $c_{1}, c_{2}$ be two children of $v_1$, where $c_{1}$ is the common neighbor of $r$ and $v_1$. Then we get an induced $P_6$, namely, $(c_{2}, v_1, c_{1}, r, c, v_2)$, where $c$ is a common neighbor of $r$ and $v_2$. (Such a $c$ must exist as ($r, v_2 \notin E(G)$).
    \end{itemize}
\end{itemize}

(Proof for the ``only if'' direction) Suppose $G$ is a connected graph with $td(G)=3$ such that $mtd(G) = 3$. We will show that $G$ cannot contain $C_4$, $P_6$, or $\graphw$ as induced subgraphs. We will prove that any connected induced subgraph $G'$ of $G$ has matched treedepth at most $3$.

Let $T$ be the matched elimination tree of $G$ of depth $3$. Let $r$ be the root of $T$. If $r$ is not in $G'$, then since $G'$ is connected, all vertices in $G'$ must come from a single sub-tree of $r$ in $T$. In that case, that sub-tree witnesses a matched elimination tree of depth at most $2$ for $G'$. If $r$ is in $G'$, then $G'$ may be obtained by deleting some level $1$ and level $2$ vertices in $T$. We will construct a matched elimination tree for $G'$ from $T$. If a level $2$ vertex is not in $G'$, we simply delete it from $T$ as well. If a level $1$ vertex $v$ is not in $G'$, then for each $u \in V(G')$ that is also a child of $v$ in $T$, there must be an edge $ru$ in $G'$ since $G'$ is connected. Therefore, we make $u$ a child of $r$ in the matched elimination tree for $G'$.
\end{proof}

By Theorem~\ref{thm:count-hom-mtd} and Theorem~\ref{thm:mtdtd}, we can conclude that every pattern $G$ with $\td(G) = 3$ has an $\otilde(m^2)$ algorithm for counting homomorphisms from $G$ to $m$-edge sparse host graphs. Therefore, the presence of $C_4$, $P_6$, or $\graphw$ as induced subgraphs in $G$ does not affect the running-time of Algorithm~\ref{alg:count-hom-mtd}. If we consider patterns $G$ with $\td(G) = 4$, then we have examples such as $P_8$ where $\td(P_8) = 4$ and $\mtd(P_8) = 5$. Therefore, we can obtain only an $\otilde(m^3)$ algorithm for counting homomorphisms from $P_8$. It would be interesting to prove theorems similar to Theorem~\ref{thm:tdmtdthree} for higher treedepth, say $4$. But, we do not even know the exact set of forbidden induced subgraphs for treedepth $4$ \cite{GT09} so this seems difficult.

\section{Matched treewidth}
\label{sec:mtw}

In this section, we introduce algorithms that count homomorphisms and subgraphs efficiently on sparse host graphs by using polynomial space. The central theorem in this section is given below.

\counthommtw*

\begin{proof}
Let $T$ be a matched tree decomposition of $G$. We fix an arbitrary assignment of edges and vertices of $G$ to bags of $T$ such that each vertex and edge is assigned to exactly one bag that contains it. Let $B$ be a bag in $T$. We consider the matched matching $M$ in $B$ as a sequence of edges $(a_1b_1$, $\dotsc$, $a_kb_k)$ by arbitrarily ordering them. Given edges $u_1v_1,\dotsc,u_kv_k$ in the host graph $H$ such that $\sigma(a_i) = u_i$ and $\sigma(b_i) = v_i$ is a valid partial homomorphism on the vertices of $B$, we define the following monomials:
\begin{align*}
 \edgemon(B, u_1v_1, \dotsc, u_kv_k) &= \prod_{e} x_{\sigma(e)} \\
 \vertexmon(B, u_1v_1, \dotsc, u_kv_k) &= \prod_v y_{\sigma(v)} 
\end{align*}
where $e$ and $v$ range over all edges and vertices assigned to $B$. 
We also define $\evmon(B, u_1v_1, \dotsc,$ $u_kv_k)$ as the product of $\edgemon(B, u_1v_1, \dotsc, u_kv_k)$ and $\vertexmon(B, u_1v_1, \dotsc, u_kv_k)$.

Let $B$ and $B'$ be bags in $T$ such that $B'$ is the parent of $B$. We arbitrarily order the vertices to obtain a sequence $X(B \cap B')$ of the vertices in $B \cap B'$ (We may assume that the vertices are labeled from $[|V(G)|]$ and choose increasing order). We define $\mapgate(B, B'$, $u_1$, $\dotsc$, $u_k)$ where $X(B \cap B') = (a_1, \dotsc, a_k)$ as the named gate which corresponds to the partial homomorphism $\sigma(a_i) = u_i, 1 \leq i \leq k$.

Algorithm~\ref{alg:hom} constructs the required circuit. Notice that all operations within the loops in Line~\ref{alg:line9},~\ref{alg:line20},~\ref{alg:line36} runs in $O(\log n)$ time. The loops itself executes for $O(m^{\lceil (t+1) / 2 \rceil})$ iterations since any matched matching on $t+1$ vertices contains at most $\lceil (t+1) / 2 \rceil$ edges. The loops in Line~\ref{alg:line2},~\ref{alg:line7}, and~\ref{alg:line18}, executes for $O(1)$ iterations since the pattern graph $G$ is fixed. Notice that we iterate over pairs $uv$ that correspond to edges instead of edges $\{u, v\}$ because the order determines the homomorphism.

When a hash table lookup for a $\mapgate(.)$ gate fails. We add it to that table with an initial value if the gate occurred on the left-hand side and replace it with $0$ if it occurs on the right-hand side.
We now prove the correctness of the circuit using parse trees. Notice that the only $+$ gates in the circuit are $\mapgate(\dotsc)$. For this proof, we can think of $r$ as $\mapgate(R, \phi)$ where $R$ is the root bag in $T$ with an empty bag as parent. Therefore, the monomials of the final polynomial correspond to the choices at these gates while building the parse tree. For some bag $B$ with parent $B'$ in $T$ and vertices $x_1,\dotsc,x_k$ in $G$, the inputs to the gate $\mapgate(B, B', x_1, \dotsc, x_k)$ correspond to valid partial homomorphisms on the vertices of $B$ that map $a_i$ to $x_i$ for all $a_i \in X(B\cap B')$. We identify these gates with the bag $B$. Now, given a homomorphism $\sigma$ from $G$ to $H$, we build its parse tree by choosing for each such gate, the restriction of $\sigma$ to the vertices in that bag. We have to prove that this choice will be consistent with the images $x_1,\dotsc,x_k$ at each gate. Indeed, this is vacuously true at the root gate (the list is empty). For an arbitrary $a \in V(G)$, consider the topmost bag $B$ where $a$ first appears. In the partial homomorphism chosen at $B$, we can freely fix the image of the vertex $a$ to that of in $\sigma$. By construction, this choice is then propagated when moving to its child gates corresponding to bags where $a$ is present (See Lines~\ref{alg:line15}, \ref{alg:line30}, and Lines~\ref{alg:line42}  ). Also, if a child bag of $B$ does not contain the vertex $a$, then the vertex $a$ will not appear in that subtree too. It is easy to see that this parse tree computes the correct monomial. Since the circuit is monotone, this proves that all monomials that correspond to homomorphisms are present in the polynomial.
For the other direction, we have to argue that only monomials that correspond to homomorphisms are present in the polynomial. Indeed, any parse tree corresponds to a sequence of choices of partial homomorphisms at each gate. We argue that these partial homomorphisms must be consistent with each other and therefore can be combined into a valid homomorphism. This is because $T$ is a tree decomposition and therefore any $a\in V(G)$ must occur in a (connected) subtree of $T$. The construction of the circuit ensures that once the image of vertex $a$ is determined in a parse tree, it is correctly propagated to all partial homomorphisms where $a$ is a member of the domain. Furthermore, since each vertex and edge in $G$ are assigned to a unique bag, their images appear exactly once in the monomial. This completes the proof.

\begin{algorithm}
\caption{Computing $\ghom_{G}[H]$}
\label{alg:hom}
\begin{algorithmic}[1]
\State Let $T$ be a near-perfect tree decomposition of $G$.

\For{each bag $B$ in $T$} \label{alg:line2}
    \For{each child $B'$ of $B$ in $T$} 
        \State Initialize empty hash table $T_{(B, B')}$.
    \EndFor
\EndFor \label{alg:line6}

\For{each non-root leaf bag $B$ in $T$} \label{alg:line7}
    \State Let $B'$ be the parent of $B$ in $T$.
    \For{$(u_1v_1,\dotsc,u_kv_k) \in E(H)^k$} \label{alg:line9}
        \State{Let $\sigma$ be $\sigma(a_i) = u_i$, $\sigma(b_i) = v_i$.} \label{alg:line10}
        \If{$\sigma$ is not a valid partial homomorphism from $G$ to $H$}
            \State{Skip this iteration}
        \EndIf
        \State Let $x_1, \dotsc, x_{k'}$ be the images of vertices in $X(B \cap B')$ in $\sigma$. 
        \State{$\mapgate(B, B', x_1,\dotsc,x_{k'}) \pluseq \evmon(B, u_1v_1, \dotsc, u_kv_k)$} \label{alg:line15}
    \EndFor
\EndFor \label{alg:line17}

\For{each non-root bag $B$ in $T$ in a bottom-up order} \label{alg:line18}
    \State Let $M = (a_1b_1,\dotsc,a_kb_k)$.
    \For{$(u_1v_1,\dotsc,u_kv_k) \in E(H)^k$} \label{alg:line20}
        \State Let $\sigma$ be $\sigma(a_i) = u_i$, $\sigma(b_i) = v_i$.
        \If{$\sigma$ is not a valid partial homomorphism from $G$ to $H$} \label{alg:line22}
            \State{Skip this iteration}
        \EndIf
      
        \State Let $B'$ be the parent of $B$ in $T$.
        \State Let $B_1,\dotsc,B_s$ be the children of $B$ in $T$.
        \State Let $x_1, \dotsc, x_{k'}$ be the images of vertices in $X(B \cap B')$ in $\sigma$.
        \State Let $w_{i1},\dotsc,w_{ik_i}$ be the images of vertices in $X(B \cap B_i)$ in $\sigma$.
        \State{$\mapgate(B, B', x_1,\dotsc,x_{k'}) \pluseq$} 
        \State{  $\evmon(B, u_1v_1, \dotsc, u_kv_k)\prod_{i=1}^{s}\mapgate(B, B_i, w_{i1},\dotsc,w_{ik_i})$} \label{alg:line30}
    \EndFor
\EndFor \label{alg:line32}

\State{Let $B$ be the root in $T$}
\State Let $B_1,\dotsc,B_s$ be the children of $B$ in $T$.

\State{$r \gets 0$}
\For{$(u_1v_1,\dotsc,u_kv_k) \in E(H)^k$} \label{alg:line36}
    \State{Let $\sigma$ be $\sigma(a_i) = u_i$, $\sigma(b_i) = v_i$.}
    \If{$\sigma$ is not a valid partial homomorphism from $G$ to $H$}
        \State{Skip this iteration}
    \EndIf
    \State Let $w_{i1},\dotsc,w_{ik_i}$ be the images of vertices in $X(B \cap B_i)$ in $\sigma$.
    \State{$r \pluseq \evmon(B, u_1v_1, \dotsc, u_kv_k)\prod_{i=1}^{s}\mapgate(B_i, B, w_{i1},\dotsc,w_{ik_i})$} \label{alg:line42}
\EndFor \label{alg:line43}
\State \Return $r$
\end{algorithmic}

\end{algorithm}

\end{proof}

\begin{remark}
    An anonymous reviewer on an earlier draft of this paper commented that the graph parameter \emph{generalized hypertree width}, denoted $\ghw$, may yield similar running-times. Indeed, we have verified that $\otilde(m^{\ghw})$-time algorithms exist \emph{and} that $\ghw \leq \lceil (\mtw + 1) / 2 \rceil$ \emph{and} the running-times coincide in the worst-case. To the best of our knowledge, the parameter $\ghw$ has not been analyzed in the context of fast algorithms for small patterns. It is primarily used for designing efficient algorithms on hypergraphs with high treewidth and low $\ghw$. We believe analyzing $\mtw$ is also useful since it is sometimes more fine-grained than $\ghw$. i.e., the class $\ghw = k$ may contain graphs from both classes $\mtw = 2k-2$ and $\mtw = 2k-1$.
\end{remark}

First, we make some fundamental observations about $\mtw$.

We can relate it to matched treedepth as we relate treewidth to treedepth.
\begin{proposition}\label{prop:mtwmtd}
  For any graph $G$, we have $\mtw(G) \leq \mtd(G) + 1$.
\end{proposition}
\begin{proof}
  Let $E$ be the optimal matched elimination tree for $G$. We construct a matched tree decomposition $T$ for $G$ from $E$ as follows: For each path from root to leaf in $E$ from the leftmost path to the rightmost path, construct bags that contain all vertices in those paths. Then, join those bags into a path by adding an edge between bags $B$ (corresponds to path to leaf $u$) and $B'$ (corresponds to the path to leaf $v$) if and only if $v$ is the next leaf in $E$ from $u$ when leaves are ordered from left to right.
\end{proof}

Now, we prove that matched treewidth cannot be much higher than treewidth.
\mtwvstw*
\begin{proof}
    Let $T$ be a tree decomposition of $G$ of width $k$. We will describe a procedure to convert $T$ to a matched tree decomposition. The construction is top-down. The final tree will have the property the vertices in each non-leaf bag will induce a perfect matching in $G$.
    
    Let $B_r$ be the root bag in $T$. The following procedure will be applied to $B_r$.
    \begin{enumerate}
    \item Find maximal matching $M$ in $B_r$.
    \item For each $v\in B_r\setminus M$ such that $N_G(v) \subseteq B_r$. Observe that since $M$ is a maximal matching $N_G(v) \subseteq V(M)$ as well. We delete $v$ from $B_r$ and add a leaf bag $B_v$ as a child of $B_r$. The bag $B_v$ will contain $v$ and the vertices in $M$. Since $v$ is adjacent to at least one vertex in $M$, this bag is matched.
    \item For each $v \in B_r\setminus M$ such that $N_G(v)\nsubseteq B_r$. We can find a $u$ such that  $u\notin B_r$ and $u\in N_G(v)$. We choose such a $u$ that is in a bag $B$ that is at a minimal distance from $B_r$ in $T$. We add $u$ to $B_r$ and all the bags in the path from $B$ to $B_r$ in $T$. We then modify $M=M\cup \{uv\}$.
\end{enumerate}
Notice that each iteration in step 2 and step 3 reduces the unmatched vertices in $B_r$ by 1 and only adds leaf bags to $T$ that are matched. In addition, if $B_r$ originally had $x$ vertices, after this procedure it will have at most $2x$ vertices as we add at most one vertex corresponding to each of the original $x$ vertices. Therefore, the size of modified $B_r$ is at most $2k+2$ and the graph induced by vertices in $B_r$ has a perfect matching.

Now, consider an arbitrary bag $B$ such that all its ancestors are bags with a perfect matching. Let $B_p$ be the parent of $B$ in $T$ and let $M$ be the perfect matching on the vertices in $B_p$. We now apply the following procedure on $B$.
\begin{enumerate}
    \item For each $v \in B \cap B_p$, let $u$ be the partner of $v$ in $M$. If $u$ is in $B$, then we match $v$ to $u$ in $B$ as well. If not, then we add $u$ to $B$ and match $v$ to $u$ in $B$. This does not violate any properties of tree decompositions. Notice that if $v$ was added to $B$ in step 3 of the procedure for the root bag, then its partner $u$ must be in $B$ as well.
    \item For $v\in B \setminus B_p$, we apply steps 2 and 3 in the procedure described for the root bag. These steps only modify the bags in the subtree of $T$ rooted at $B$ and will only add to $T$ leaf bags that are matched. Moreover, at the end of this step, the vertices in bag $B$ induce a graph that has a perfect matching.
\end{enumerate}

Observe that the size of a bag $B$ can at most double from its original size since we add at most one vertex for each vertex originally in the bag. This is true for all the newly added leaf bags as well. Therefore, we have constructed a matched tree decomposition of width at most $2k+1$.
\end{proof}

We now show an explicit family that has close to the worst-case relation between treewidth and matched treewidth.

\mtwknn*
\begin{proof}
	Let $T$ be just an edge with vertex set $\{K_{n,n} \setminus \{a\},K_{n,n} \setminus \{a'\}\}$ for some $a$ and $a'$ that are in the same part. One can easily verify that $T$ is matched tree decomposition of $K_{n,n}$. Thus $mtw(K_{n,n}) \leq 2n-2$.

	Let $T$ be an arbitrary matched tree decomposition of $K_{n, n}$. Root $T$ at a leaf bag, say $X$ and let $X_1$ be the only child of $X$. If $X \subseteq X_1$, then we can delete $X$ from $T$. Therefore, there exists a vertex $u \in X$ of $K_{n,n}$ such that $u \not\in X_1$. So $N(u) \subseteq X$. Assume wlog that $u \in A$. Then, we get $B \subseteq X$. Now to match the $n$ vertices in $B$, we need at least $n-1$ vertices from $A$ in the bag $X$. So $mtw(K_{n,n}) \geq 2n-2$.
\end{proof}

It is easy to see that $\tw(K_{n, n}) = n$. Therefore, its matched treewidth is only 3 less than the worst case $2n+1$.

We now use the Algorithm~\ref{alg:hom} and Definition~\ref{def:spasm} to count paths in cycles in sparse host graphs. Instead of analyzing all graphs in $\spasm(P_{10})$, we completely characterize the matched treewidth of graphs with treewidth 2. This will simplify the case analysis required for proving algorithmic upper bounds for many small pattern graphs.

\subsection{Matched treewidth of partial 2-trees}

In this section, we study the matched treewidth of partial 2-trees. A summary is given in Table~\ref{tab:mtw2tree}. From Theorem~\ref{thm:mtwvstw}, we have $\mtw(G) \leq 5$ when $\tw(G) = 2$ which yields the last row. The graph $Y$ (See Figure~\ref{fig:Y}) satisfies $\tw(Y) = 2$ and $\mtw(Y) = 5$. But $Y$ is also an induced subgraph of $Z$ (See Figure~\ref{fig:Z}) which satisfies $\tw(Z) = 2$ and $\mtw(Z) = 4$. Therefore, a forbidden induced minor subgraph characterization is not applicable for this case. We now prove the remaining two characterizations.

\begin{table}[ht]
    \centering
    \begin{tabular}{rl}
         $\mtw \leq .$ & Forbidden induced minor \\
         \hline
         2 & $C_5$ \\
         3 & $X$ \\
         4 & Not applicable. \\
         5 & None. 
    \end{tabular}
    \caption{Matched treewidth of partial 2-trees.}
    \label{tab:mtw2tree}
\end{table}

\twmtwtwo*

\begin{proof}
(Proof for the ``if'' direction) Suppose for contradiction that there is a partial 2-tree $G$ such that $\mtw(G) = 2$ and $G$ has a $C_5$-induced-minor. Let $T$ be a matched tree decomposition of $G$ with width $2$. Since $C_5$ is an induced minor, we can obtain $C_5$ from $G$ by deleting vertices and contracting edges. For all $v$ that is deleted, delete $v$ from all bags of $T$. Similarly, for all edges $uv$ that are contracted, replace $u$ and $v$ consistently in all the bags of $T$ by one of $u$ or $v$. We obtain a (not necessarily matched) tree decomposition $T'$ of $C_5$. Assume wlog that $T'$ is a reduced tree decomposition. Since $T$ has width at most $2$, any bag in $T'$ contains at most $3$ vertices. If a bag in $T'$ has $3$ vertices, then they already form a $P_3$ since it must have been so in $T$. There cannot be a bag of one vertex in $T'$ because it is reduced and $C_5$ is connected. We claim that a bag of size $2$ in $T'$ must contain some $u$ and $v$ such that $uv \in E(C_5)$.
\begin{claim}
  Let $u$ and $v$ be two vertices in $C_5$ that are not adjacent. Then, the tree decomposition $T'$ cannot contain the bag $\{u, v\}$.
\end{claim}
\begin{proof}
Suppose for contradiction that $T'$ contains the bag $B = \{1, 3\}$. The other cases are symmetric. Root the tree $T'$ at that bag. We now analyze various cases.
\begin{enumerate}
    \item ($B$ has one child) Since $1$ and $3$ are adjacent to other vertices, the child of $B$ must contain both $1$ and $3$, contradicting the fact that $T'$ is reduced.
    \item ($B$ has more than one child) We split this case sub-cases.
    \begin{enumerate}
        \item (The edges $15$ and $34$ are covered in distinct subtrees (Say $T_1$ and $T_2$)  of $B$) Since $45 \in E(C_5)$, this edge must be covered. The bag that covers $45$ must have a path containing $5$ to the bag covering $15$ in $T_1$ due to the connectivity of $5$. Similarly, this bag must have a path containing $4$ to the bag covering $34$ in $T_2$ due to the connectivity of $4$. But, this is impossible since $B$ contains only $1$ and $3$.
        \item (The edges $15$ and $34$ are covered in the same subtree of $B$) This case is split into further sub-cases.
        \begin{enumerate}
            \item (The bag $B$, and the bags containing $15$ and $34$ occur on the same path in $T'$) Suppose the path is from $B$ to the bag containing $3$ and $4$ via the bag $B'$ containing $1$ and $5$. The other case is similar. Then, the bag $B'$ must contain $3$ to maintain the connectivity of $3$. Now, we have $B' \supset B$, a contradiction. 
            \item (The bags containing $15$ and $34$ has a common ancestor that is a proper descendant of $B$) Let $B'$ be this common ancestor. This bag $B'$ must contain $3$ as it lies on the path from $B$ to the bag containing $3$ and $4$. Also, this bag $B'$ must contain $1$ as this lies on the path from $B$ to the bag containing $1$ and $5$. But, then $B' \supseteq B$, a contradiction.
        \end{enumerate}
    \end{enumerate}
\end{enumerate}
\end{proof}
This completes the proof of the "if" direction.

(Proof of the ``only if'' direction) We use proof by contradiction. Suppose $G$ is a counter-example on minimum number of vertices.

\begin{claim}\label{claim:G2-conn}
    The graph $G$ is $2$-connected.
\end{claim}
\begin{proof}
Suppose $G$ has a cut vertex $v$. By deleting $v$, we obtain smaller graphs $G_1, \dotsc, G_m$ for some $m > 1$. Any cycle in $G$ is also a cycle in $G_i$ for some $i$. Since $G$ is minimal, each $G_i$ has a matched tree decomposition of width at most $2$, say $T_i$. For all $i$, let $B_i$ be a bag in $T_i$ that contains $v$. Add edges between $B_1$ and $B_i$ for all $1 < i \leq m$. This is a matched tree decomposition for $G$ of width $2$, a contradiction.
\end{proof}

We take a 2-tree $G'$ that is a super-graph of $G$ and has the same vertex set as $G$.

\begin{claim}\label{claim:common}
  Let $uv$ be an edge in $G'$ but not in $G$. Then $u$ and $v$ have a common neighbor in $G$.
\end{claim}
\begin{proof}
Since $G$ is 2-connected, there are two internally disjoint paths $P = uu_1\dotsm u_kv$ and $P' = uv_1\dotsm v_\ell v$ in $G$ between $u$ and $v$. We may assume that $P$ and $P'$ are induced paths in $G$. If $k$ or $\ell$ is $1$, then $u$ and $v$ has a common neighbor. So we assume $k$ and $\ell$ are at least $2$. Now $PP'$ is a cycle of length at least $6$ and it must have a chord (otherwise, this is a $C_5$-induced-minor in $G$.). Therefore, there is some $i$ and $j$ such that $u_i$ is adjacent to $v_j$. the edges on $P$, $P'$, together with this chord is a $K_4$-minor in $G'$, a contradiction.
\end{proof}

Let $T$ be a standard tree decomposition of $G'$. This is a matched tree decomposition for $G$. We show that this is also a matched tree decomposition for $G$. It is enough to show that a set $\{u, v, w\}$ which forms a triangle in $G'$ induces either a $P_3$ or a triangle in $G$. Assuming the contradiction, we get that $v$ (say) is not adjacent to both $u$ and $w$. By Claim~\ref{claim:common}, we obtain vertices $x$ that is a common neighbor of $u$ and $v$ in $G$, and $y$ that is a common neighbor of $v$ and $w$ in $G$. Since $G'$ is $K_4$-minor-free, we have $x \neq y$, $wx \not\in E(G)$, $xy \notin E(G)$, and $uy \not\in E(G)$. If $uw \in E(G)$, then $uwyvxu$ is an induced $C_5$ in $G$, a contradiction. So $uw \not\in E(G)$, and by Claim~\ref{claim:common}, we obtain a vertex $z$ that is a common neighbor of $u$ and $w$ in $G$. Again, since $G'$ is $K_4$-minor-free, we have $z \neq x$, $z \neq y$, $zx \not\in E(G)$, $zv \notin E(G)$, and $zy \not\in E(G)$. Then $uzwyvxu$ is an induced $C_6$ in $G$, a contradiction. This completes the proof.
\end{proof}

We now show that the graph $X$ (See Figure~\ref{fig:x}) is a partial 2-tree that has matched treewidth more than 3.
\begin{lemma}
    $mtw(X)\geq 4$
    \label{xmtw}
\end{lemma}
\begin{proof}
    Consider a reduced, matched tree decomposition $T$ of $X$. Suppose for contradiction that $X$ has width strictly less than $4$. Consider a leaf bag $B_l$ in $T$. The bag $B_l$ must contain $u_0$ or $u_2$ or $u_4$ as all edges in $X$ are incident on one of these vertices. All cases are symmetric so wlog, we can assume $B_l$ is a super-set of $\{u_0,u_1\}$. Consider the tree $T$ to be rooted at $B_l$. As neither $u_0$ nor $u_1$ is pendant and $T$ is reduced, it follows that $B_l$ must contain at least one other vertex. Let $B_p$ be the bag adjacent to $B_l$. We claim that $\{u_0, u_2\}$ or $\{u_0, u_4\}$ is contained in $B_p$. If $u_0 \notin B_p$, then $\{u_1,u_1',u_5,u_5'\} \subseteq B_l$ and therefore $|B_l| \geq 5$. Suppose $u_1 \in B_p$. Since $T$ is reduced, this means that some vertex other than $u_0$ and $u_1$ was in $B_l$ but not in $B_p$. If this vertex is $u_2$ or $u_4$, then $|B_l| \geq 5$. If it is one of the other vertices, either $u_2$ or $u_4$ is in $B_l$ and $B_p$. If $u_1 \notin B_p$, then $B_l$ must contain $u_2$ and therefore so must $B_p$. Now, we assume wlog that both $u_0$ and $u_2$ are in the bag $B_p$. If $B_p$ also contains $u_4$, then we are done as the size of matched bag would be at least $5$.
    
    Let $B_c$ be a descendant bag of $B_p$ in $T$ such that $B_c$ is the root of the subtree of $T$ that contains $u_4$. Again, if $\{u_0, u_2, u_4\} \subseteq B_c$, then $|B_c| \geq 5$ and we are done. So assume wlog that $u_2 \notin B_c$. If $u_0 \notin B_c$, then $B_c \supseteq \{u_5, u_5', u_3, u_3'\}$ and we are done. This is because these vertices are common neighbors of $u_4$ with $u_0$ or common neighbors of $u_4$ with $u_2$ and $B_c$ is the root bag of the sub-tree where $u_4$ appears in $T$. So $u_0$ is also in $B_c$. But $B_c$ also contains $u_3$ and $u_3'$ and $B_c$ must contain at least one more vertex to match $u_0$.
\end{proof}

We now prove that forbidding $X$-induced-minor (for partial $2$-trees) exactly gives us the class of graphs with $\mtw(.) \leq 3$.

\begin{definition}
 A minimum chordal completion $\Tilde{G}$ of some graph $G$ is said to be \emph{special} if no independent set of size $3$ in $G$ induces a clique in $\Tilde{G}$. We write smcc instead of special minimum chordal completion.
\end{definition}

Note that treewidth of a graph is same as the treewidth of its smcc (if it exists).
\xminorfree*
\begin{proof}
  (Proof for the ``if'' direction) Let $G$ be a graph with a matched tree decomposition $U$ of width $3$. Suppose for contradication that $G$ has an $X$-induced-minor. Then, we can also obtain $X$ by contracting some edges of an induced subgraph of $X'$ of $G$. Let $T'$ be the tree decomposition of $X'$ obtained from $U$ by deleting vertices in the set $S$ from all bags. Since $X$ can be obtained from $X'$ by contracting some edges, it will have a set of three vertices, say $u_0, u_2, u_4$, that form an independent set and internally vertex-disjoint paths $Q_{i+1}, Q'_{i+1}$ from $u_i$ to $u_{i+2}$ \footnote{All indices $i$ for $u_i$ and $u_i'$ in this proof are modulo $6$.}. We define the map $f: V(X') \mapsto V(X)$ is such that $f(u_i) = u_i$ for $i \in \{0, 2, 4\}$, $f$ maps all internal vertices in the path $Q_{i+1}$ to $u_{i+1}$ and $Q'_{i+1}$ to $u'_{i+1}$. The map $f$ corresponds to the edge contraction that we use to obtain $X$ from $X'$. Let $T$ be the tree-decomposition of $X$ obtained from $T'$ by applying $f$ to all vertices in all bags. We can remove bags of size $0$ and size $1$ from $T$ using standard techniques.
  
  Now we modify $T$ to obtain a matched tree decomposition of width $3$ for $X$ thereby deriving a contradiction. Every bag $B$ in $T$ has a corresponding bag $P$ in $T'$. Note that if $uv$ is an edge in $X'$, then $f(u)f(v)$ is an edge in $X$, whenever $f(u) \neq f(v)$. If $B$ has size $4$, then since $U$ is a matched tree decomposition, the bag $P$ must also be matched. Therefore, if all bags in $T'$ has size $4$, we are done. Otherwise, there is a bag $B$ in $T$ that has size less than three. We root the tree $T$ at $B$ and modify $T$ in a top-down fashion. We need the following claim to prove the correctness of this procedure.
  
  \begin{claim}
      Let $B'$ be some bag in $T$ of size $3$, then the vertices in $B'$ cannot form an independent set of size $3$ in $X$.
  \end{claim}
  \begin{proof}
    If the corresponding bag $P$ to $B'$ in $U$ has size $4$, then since $U$ is a matched tree decomposition, the bag $P$ is matched. So $B'$ was obtained by the deletion of a vertex or the contraction of an edge. Both of which will retain at least one edge in the bag.
    
    If the corresponding bag $P$ has size $3$, then since $U$ is a matched tree composition, the subgraph induced by $B'$ has $P_3$ as a subgraph.
  \end{proof}
  
  So we can assume the following:
  \begin{enumerate}
      \item Every bag in $T$ of size $4$ is matched.
      \item Every bag in $T$ of size $3$ has at least one edge.
      \item There are no bags in $T$ of size $0$ or size $1$.
  \end{enumerate}
  
  We now process a bag $B$, starting from the root, assuming that all ancestors of $B$ are already matched. We split the procedure into two cases.
  
  ($|B| = 3$) Let $B = \{u, v, w\}$, $uv$ is an edge, and $w$ is not adjacent to $u$ or $v$. If $w$ is present in the parent of $B$, then it is matched to some $w'$. We add $w'$ to $B$ as well. If $w$ is not present in the parent, then we choose a nearest descendant $B'$ of $B$ in $T$ that contains a neighbor $w'$ of $w$. Every bag in the path from $B$ to $B'$ contains $w$. We add $w'$ to every bag in this path. Now, the bag $B$ is matched. Also, every other bag that we modified (in the path from $B$ to $B'$) now has size $3$ and has at least one edge or has size $4$ and is matched. So all properties are preserved.
  
  ($|B| = 2$) Let $B = \{u, v\}$ and $uv$ is not an edge. If $u$ is present in $B$'s parent in $T$, then it is matched to some $u'$ there. We add $u'$ to $B$. Otherwise, we choose a nearest descendant $B'$ of $B$ in $T$ that contains a neighbor $u'$ of $u$. We add $u'$ to all bags in the path from $B$ to $B'$. $B$ is now a bag of size $3$ with at least one edge. Also, every other bag that we modified now has size $3$ and has at least one edge or has size $4$ and is matched. Now, if $B$ is matched we are done. Otherwise, we apply the previous case to $B$.
  
  (Proof for the ``only if'' direction) We now prove a series of lemmas that will prove the theorem.

\begin{lemma}\label{lem:IntDisPath}
Let $G$ be a graph with $tw(G) \leq 2$ and $\Tilde{G}$ is a smcc of $G$. If there exists three internally disjoint paths from a vertex $u$ to $v$ in $G$, then $uv$ is an edge in $\Tilde{G}$ 
\end{lemma}

\begin{proof}
Suppose for contradiction $u$ is not adjacent to $v$.
Let $Q_1,Q_2,Q_3$ be three internally disjoint paths from $u$ to $v$. Let $Q_1',Q_2'$ be the shortest path from $u$ to $v$ in $\Tilde{G}[V(Q_1)]$ and $\Tilde{G}[V(Q_2)]$, respectively. 
Note that both of them are of length at least $3$. Since $\Tilde{G}$ is chordal, the cycle obtained from $Q_1'$ and $Q_2'$ has a chord say $xy$. Note that $\{x,y\} \cap \{u,v\} = \emptyset$. Without loss of generality, we may assume that $x$ is in $Q_1'$. So $y$ is in $Q_2'$. Therefore, $Q_1',Q_2'$ and $Q_3$ gives a minor of $K_4$ in $\Tilde{G}$. This is a contradiction.
\end{proof}

\begin{lemma} \label{lem:MFThenSmcc}
 Let $G$ be a $X$-induced-minor-free connected graph with tree-width $2$. Then, there exists a smcc $\Tilde{G}$ of $G$. 
\end{lemma}

\begin{proof}
Suppose for contradiction that there exists a minimal counter-example $G$ with $\tw(G) \leq 2$ that is $X$-induced-minor-free but does not have a smcc. Since $G$ is minimal, it has to be biconnected and by the characterization of treewidth $2$ graphs, we can assume that $G$ is series-parallel. We will use the series-parallel graph characterization to prove that $G$ does not exist.

Let $s$ and $t$ be the source and terminal of $G$. Since $G$ is a counter example, it is not an edge. Suppose $(G,s,t)$ is the series composition of $(G_1,s_1,t_1)$ with $(G_2,s_2,t_2)$ ($t_1$ is identified with $s_2$). Since $G$ is a minimal counter example, there exists smccs $\Tilde{G_1}$ of $G_1$ and $\Tilde{G_2}$ of $G_2$. Let $\Tilde{G}$ be the graph obtained from $\Tilde{G_1}$ and $\Tilde{G_2}$ by identifying $t_1$ with $s_2$. 
The graph $\Tilde{G}$ is a smcc of $G$, a contradiction. Hence, the graph $(G,s,t)$ is the parallel composition of strictly smaller graphs $(G_1,s_1,t_1)$ and $(G_2,s_2,t_2)$.

\begin{claim}\label{claim:not-edge-no-edge}
    Let $G$ be the parallel composition of two smaller graphs $G_1$ and $G_2$. Furthermore, assume that $G_1$ is the series composition of two or more graphs $G_i'$ for $1 \leq i \leq m$. Then, for all $1 \leq i \leq m$, if $G_i'$ is not an edge, then $s_i't_i'$ is also not an edge.
\end{claim}
\begin{proof}
    Suppose there exists $1 \leq i \leq m$ such that $G_i'$ is not an edge but $s_i't_i'$ is an edge in $G_i'$. So $G_i'$ is of order at least $3$. There exists a smcc $G_i$ of the graph obtained from $G$ by deleting all the internal vertices of $G_i'$ (Since $G$ is a minimum counter example). Again there exists a smcc $\Tilde{G_i'}$ of $G_i'$. Define $\Tilde{G}$ to be the graph obtained from $G_i$ and $\Tilde{G_i'}$ by identifying the edge $s_i't_i'$. The graph $\Tilde{G}$ is a smcc of $G$. 
\end{proof}

Note that at least one of $G_1$ or $G_2$ is not an edge. We now split the proof into two exhaustive cases:

\begin{itemize}
    \item (Case 1) Either $G_1$ or $G_2$ is an edge.

    Without loss of generality we may assume that $G_2$ is an edge. So $s_1t_1$ is not an edge in $G_1$ (as otherwise $G_1=G$). Suppose $(G_1,s_1,t_1)$ is a parallel composition of $(G_3,s_3,t_3)$ and $(G_4,s_4,t_4)$. Since $s_1t_1 \not\in E(G_1)$, we have $s_3t_3 \not\in E(G_3)$ and $s_4t_4 \not\in E(G_4)$. Therefore, both $G_3$ and $G_4$ has at least one internal vertex. Let $G_3'$ be the graph obtained from $G_3$ by adding the edge $s_3t_3$ and $G_4'$ be the graph obtained from $G_4$ by adding the edge $s_4t_4$. Since $G$ is a minimum counter example, there exists smccs $\Tilde{G'_3}$ of $G'_3$ and $\Tilde{G_4'}$ of $G_4'$. Note that $G$ is obtained from $G_3'$ and $G_4'$ by identifying edges $s_3t_3$ and $s_4t_4$. So the graph $\Tilde{G}$ obtained from $\Tilde{G_3'}$ and $\Tilde{G_4'}$ by identifying those edges is a smcc of $G$. This is a contradiction.
    
    We may assume that $(G_1,s_1,t_1)$ is a series composition of smaller graphs $(G_1',s_1',t_1')$, $(G_2',s_2',t_2')$, $\dotsc$, $(G_m',s_m',t_m')$ for some $m \geq 2$ such that for all $2 \leq i \leq m$, the vertex $t_{i-1}'$ is identified with $s_i'$ and for all $1 \leq i \leq m, (G_j',s_j',t_j')$ is either an edge or a parallel composition of two smaller graphs.

Note that $G$ is not a cycle since $G$ is a counter example. So, there exists some $i$ such that the graph $G_i'$  is not an edge. Therefore, by Claim~\ref{claim:not-edge-no-edge}, we have that $s_i't_i'$ is not an edge.

\begin{claim}\label{claim:noijk}
There does not exist $1 \leq j (\neq i) < k (\neq i) \leq m$ such that $s_j't_j', s_k't_k'$ are non-edges. 
\end{claim}
\begin{proof}
Suppose $s_j't_j', s_k't_k'$ are non-edges, for some $1 < j (\neq i) < k (\neq i) \leq m$. Without loss of generality, we may assume that $i < j$. So $G_i'$ is a parallel composition of two smaller graphs. Hence there exist two internally disjoints paths $Q_1,Q_1'$ (each of length at least $3$), from $s_i'$ to $t_i'$ in $G_i'$. Similarly there exist two internally disjoints paths $Q_2,Q_2'$ (each of length at least $3$), from $s_j'$ to $t_j'$ in $G_j'$. Also there exist two internally disjoints paths $Q_3,Q_3'$ (each of length at least $3$), from $s_k'$ to $t_k'$ in $G_k'$. Again, there exists a path from $t_i'$ to $s_j'$, through $G_{i+1}', \dots G_{j-1}'$, a path from $t_j'$ to $s_k'$ through $G_{j+1}', \dots G_{k-1}'$. Again we have a path from $t_k'$ to $s_i'$ through $G_{k+1}', \dots G_m', {s_m't_1'}, G_1', \dots G_{i-1}'$. This gives a subdivision of $X$. This a contradiction. Hence, the above claim is true. 
\end{proof}

\begin{claim}
\label{clm:EdgeInSCC}
There exists a smcc of $G_i'$ such that $s_i't_i'$ is an edge in the completion. 
\end{claim}

\begin{proof}
Since $s_i'$ is not adjacent to $t_i'$, $(G_i',s_i',t_i')$ is parallel composition of two smaller graphs (By Claim \ref{claim:not-edge-no-edge}), say $(G_3,s_3,t_3)$ and $(G_4,s_4,t_4)$. Let $G'$ be the graph obtained from $G$ by contracting the edge $st$. This is an $X$-induced-minor-free graph. The minimality of $G$ says that $G'$ has an smcc, say $\Tilde{G'}$. By Lemma \ref{lem:IntDisPath}, $s_i'$ is adjacent to $t_i'$ in $\Tilde{G'}$. Note that the graph induced by the vertices of $G_i'$ in $\Tilde{G'}$ is a smcc of $G_i'$ in which $s_i't_i'$ is an edge.

\end{proof}

By Claim~\ref{claim:noijk}, the following cases are now exhaustive.

\begin{itemize}
    \item (Case 1a) $s_j't_j'$ is an edge, for all $1 \leq j (\neq i) \leq m$.
    
    By the above claim, there exists a smcc $\Tilde{G_i'}$ of $G_i'$ in which $s_i't_i'$ is an edge. We construct an smcc of $G$ as follows: It contains the smcc $\Tilde{G_k'}$ for all $G_k'$ and the edge $st$. We then add the edges $ss_i'$ and $s_i't_j'$ for all $1 \leq j (\neq i) \leq m$. These edges ensure that there are no chordless cycles through $G_1$ and $G_2$. This smcc gives us a contradiction. 

    \item (Case 1b) There exists $1 \leq j (\neq i) \leq m$ such that $s_i'$ is not adjacent to $t_j'$.
    
    Without loss of generality we may assume that $i<j$. By the claim \ref{clm:EdgeInSCC}, there exists a smcc $\Tilde{G_i'}, \Tilde{G_j'}$ of $G_i'$ and $s_j'$, respectively, such that $s_i't_i' \in E(\Tilde{G_i'})$ and $s_j't_j' \in E(\Tilde{G_j'})$. We construct an smcc of $G$ as follows: It contains the smcc $\Tilde{G_k'}$ for all $G_k'$ and the edge $st$. We then add edges $t_i's_k'$ for all $1 \leq k < i$, $st_k'$ for all $i < k < j$, and $s_j't$ and $s_j't_k'$ for all $j < k \leq m$. This smcc gives us a contradiction.    
\end{itemize}

The above constructions are illustrated in Figure~\ref{fig:xcase1}. The edges colored black are the edges in $G$. The red colored edges are those added to construct the smcc.

\begin{figure}[ht!]
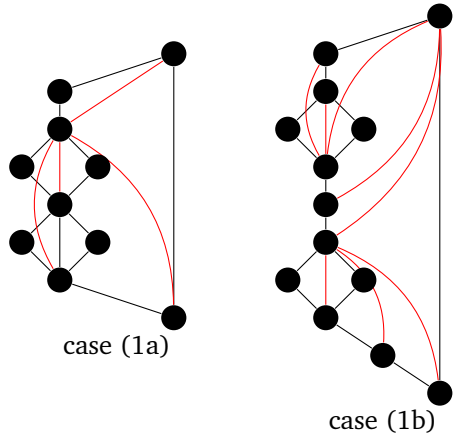

  \ctikzfig{xcase1}
  \caption{Constructing smcc in case 1.\label{fig:xcase1}}
\end{figure}

\item (Case 2) Both $G_1$ and $G_2$ are not edges. Suppose $G_1$ and $G_2$ both are parallel composition of smaller graphs. So there exists two internally disjoints paths $P,P'$ from $s_1$ to $t_1$ in $G_1$ and two internally disjoints paths $Q,Q'$ from $s_2$ to $t_2$ in $G_2$. Let $G'$ be the graph obtained from $G_1$ by adding $Q$. Since, $G$ is a minimum counter example, $G'$ has a smcc $\Tilde{G'}$. Again, $P,P',Q$ are three mutually internally disjoint paths from $s_1$ to $t_1$. By, lemma \ref{lem:IntDisPath}, $s_1t_1$ is an edge in $\Tilde{G'}$. The graph $\Tilde{G_1}$ induced by $V(G_1)$ in $\Tilde{G'}$ is a smcc such that $s_1t_1$ is an edge in $\Tilde{G_1}$. Similarly, we can show that there exists a smcc $\Tilde{G_2}$ of $G_2$ such that $s_1t_2$ is an edge in $\Tilde{G_2}$. The graph obtained from $\Tilde{G_1}$ and $\Tilde{G_2}$ by identifying the edges $s_1$ with $s_2$ and $t_1$ with $t_2$ is a smcc of $G$.

We may assume that at least one of $G_1$ or $G_2$ is a series composition of graphs. We decompose $G_1$ and $G_2$ into series composed graphs repeatedly, until we cannot. i.e., all the individual components are parallel compositions or an edge. Let $(G_1,s_1,t_1)$ be series compositions of $(G_1',s_1',t_1'), (G_2',s_2',t_2'), \dotsc, (G_m',s_m',t_m')$ and let $(G_2,s_2,t_2)$ be series compositions of $(G_{m+1}',s_{m+1}',t_{m+1}'), (G_{m+2}',s_{m+2}',t_{m+2}'), \dotsc, (G_{m+\ell}',s_{m+\ell}',t_{m+\ell}')$ in that order.

\begin{claim}
    There exists a $1 \leq i \leq m+\ell$ such that $s_i't_i'$ is an edge. 
\end{claim}
\begin{proof}
Suppose for contradiction that $s_i'$ is not adjacent to $t_i'$, for all $1 \leq i \leq m+\ell$. So, the graph $(G_1',s_1',t_1')$ is a parallel composition of two smaller graphs. Hence, there exists two internally disjoint paths, each of length at least $3$, from $s_1'$ to $t_1'$ in $G_1'$. Similarly, there are two internally disjoint paths, each of length at least $3$, from $s_i'$ to $t_i'$ in $G_i'$, for all $1 \leq i \leq m+\ell$. Since $m+\ell >2$, we get an $X$-induced-minor in $G$, a contradiction.
\end{proof}

Thus, there exists a $1 \leq i \leq m+\ell$ such that $s_i't_i'$ is an edge. By Claim~\ref{claim:not-edge-no-edge}, we can conclude that $G_i'$ is an edge. Without loss of generality, we may assume that $1 \leq i \leq m$. Note that $(G \setminus \{s_i't_i'\}, s_i',t_i')$, is the graph obtained by the series composition of $G_{i-1}',G_{i-2}', \dotsc, G_1'$, $G_2, G_m', G_{m-1}', \dotsc, G_{i+1}'$ by identifying $s_{i-k}',s_1', t_2,t_j'$ with $t_{i-k+1}', s_2,t_m'$ and $s_j'$, respectively, for $1 \leq k < i$ and $i < j \leq m$. Again, the graph $(G,s_i',t_i')$ is the parallel composition of $(G \setminus \{s_i't_i'\}, s_i',t_i')$ with an edge. Now, we can use case (1) to get a contradiction. Hence, a minimum counter example does not exist.
\end{itemize}
\end{proof}

\begin{lemma}\label{obs:hasEdge}
Let $G$ be a graph with treewidth $k$ and $T$ be a tree decomposition of $G$ such that no bag of size $k+1$ of $T$ is independent in $G$. Then, $mtw(G) \leq 2tw(G)-1$.
\end{lemma}
\begin{proof}
  Observe that in the proof of Theorem~\ref{thm:mtwvstw}, the construction yields a matched tree decomposition of width $2\tw(G) - 1$ if each bag of size $\tw(G) + 1$ has at least one edge.
\end{proof}
By Lemma~\ref{lem:MFThenSmcc}, every $G$ such that $\tw(G) = 2$ and $G$ is $X$-induced-minor-free has an smcc. Now, a standard tree decomposition $T$ of $\Tilde{G}$ is also a tree decomposition of $G$ with the additional property that every bag of size $3$ has at least one edge. We apply Lemma~\ref{obs:hasEdge} to construct a matched tree decomposition of width $3$ for $G$.
\end{proof}

To complete the characterization of matched treewidth of partial 2-trees, we prove the following lemma.

\begin{figure}[ht!]
\ctikzfig{Y}
\caption{The graph $Y$.}
\label{fig:Y}
\end{figure}

\begin{lemma}
   $mtw(Y) \geq 5$
\end{lemma}
\begin{proof}
    The proof is similar to the proof of Lemma~\ref{xmtw}. We additionally use the fact that $u_1$, $u_2$, and $u_3$ do not have any common neighbors. Suppose for contradiction that $Y$ has a reduced, matched tree decomposition $T$ of width strictly less than $5$. Let $B_l$ be a leaf in $T$. We first argue that $B_l$ must contain one of $u_1$, $u_2$, or $u_3$. If not, it contains only a subset of the other edges, say, like $u_4u_5$ (other cases are symmetric). Since those vertices are not pendant, the neighbor of $B_l$ in $T$, which we call $B_p$, will be a superset of $B_l$. Suppose $B_l$ contains $u_1$ and $u_4$ (rest of the cases are symmetric). Since $T$ is reduced and has width less than $5$, the bag $B_p$ must contain both $u_1$ and $u_2$ (or $u_1$ and $u_3$, a symmetric case). Root $T$ at $B_l$.
    
    If $u_3 \in B_p$, we are done since $|B_p| \geq 6$. Let $B_c$ be the closest descendant of $B_p$ that contains $u_3$. Assume wlog that $u_2 \notin B_c$. Both $u_1$ and $u_2$ cannot be missing from $B_c$. In that case, all neighbors of $u_3$ must be in $B_c$ and that will imply $|B_c| \geq 6$. So $u_1$ and $u_3$ are both in $B_c$ (the other case is symmetric). It must also contain a vertex $v$ that is a neighbor of $u_1$ since it is matched. Consider the path $P$ from $B_p$ to $B_c$ in $T$. If $u_8$ ($u_{10}$) does not appear in this path, then $u_9$ ($u_{11}$ resp.) must appear on all bags in this path. Therefore, $B_p$ will have to contain $u_3$ to match $u_9$ (or $u_{11}$). So $u_8$ and $u_{10}$ must appear in this path.
    
    Let $B_1$ be the first bag in $P$ where both $u_8$ and $u_{10}$ have appeared. We split the proof into two cases.
    
    (Both $u_8$ and $u_{10}$ appear in $B_1$) If $u_2 \notin B_1$, then $|B_1| \geq 6$ and we are done. Otherwise, $B_1 = \{u_1, v, u_2, u_8, u_{10}\}$ for some neighbor $v$ of $u_1$. Let $B_2$ be the first bag from $B_1$ to $B_c$ where the edge $u_8u_9$ or $u_{10}u_{11}$ appears. Then, either $B_2 = \{u_1, v, u_2, u_{10}, u_8, u_9 \}$ ($u_{10}u_{11}$ has not appeared) or $B_2 = \{u_1, v, u_{10}, u_{11}, u_8, u_9\}$ (both edges appear simultaneously) and we are done.
    
    (One of $u_8$ or $u_{10}$ is missing from $B_1$) This means either $u_9$ or $u_{11}$ is in $B_1$. We have either $B_1 = \{u_1, v, u_2, u_{10}, u_9, u_3\}$ or $B_1 = \{u_1, v, u_2, u_8, u_{11}, u_3\}$ and we are done.
\end{proof}

Figure~\ref{fig:Z} is a super-graph of $Y$ that has lower matched treewidth.

\begin{figure}[ht!]
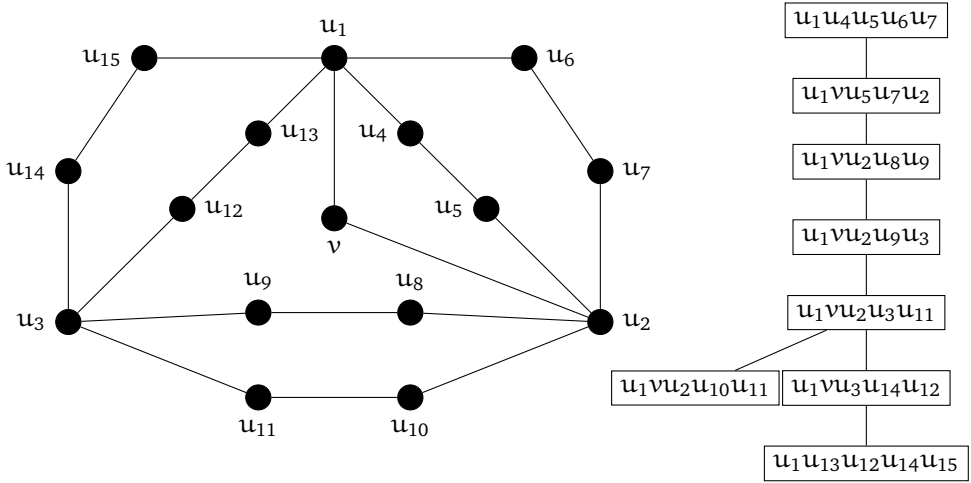

\ctikzfig{Z}
\caption{The graph $Z$ and its matched tree decomposition of width $4$.}
\label{fig:Z}
\end{figure}

\subsection{Finding and counting Subgraphs and induced subgraphs using matched treewidth}

\pkcountmtw*
\begin{proof}
There are more than 300 graphs in $\spasm(P_{10})$. We have verified that all of them have $\mtw$ at most $3$. To minimize the work, we can filter out all graphs in the spasm that has $\tw(G') = 1$ or $\tw(G') = 2$ and $G'$ is $X$-induced-minor-free. Observe that since $X$ has $9$ vertices $12$ edges, it cannot be an induced minor in any of the graphs in $\spasm(P_{10})$. Also, none of the forbidden minors for treewidth $4$ can appear in $\spasm(P_{10})$. Therefore, we only need to analyze graphs of treewidth $3$ in $\spasm(P_{10})$. There are only $18$ such graphs. They are listed in a pdf file in the repository associated with this paper (\url{https://github.com/anonymous1203/Spasm}).

Since $\spasm(P_k) \subseteq \spasm(P_{k+1})$ for all $k \geq 2$, we can make the same claim for all paths on fewer than $10$ vertices. We now use Equation~\ref{eq:subcount} to compute the result.
\end{proof}

\begin{remark}
Since $K_5 \in \spasm(P_{11})$, and treewidth of $K_5$ is $4$, the above method cannot yield an $\otilde(m^2)$ algorithm for $P_k$ where $k \geq 11$.
\end{remark}

We note that the above proof also yields $\otilde(m^2)$ time algorithms for counting any pattern in $\spasm(P_{10})$. This is because if $G \in \spasm(P_{10})$, then $\spasm(G) \subseteq \spasm(P_{10})$. In the following theorem, we point out an important class of graphs in $\spasm(P_{10})$.
\ckcountmtw*
\begin{proof}
We observe that $\spasm(C_{k}) \subset \spasm(P_{10})$ for $k \leq 9$.
\end{proof}

\begin{remark}
Since $K_5 \in \spasm(C_{10})$, and treewidth of $K_5$ is $4$, the above method cannot yield an $\otilde(m^2)$ algorithm for $C_k$ where $k \geq 10$.
\end{remark}

We can also use our efficient construction of arithemetic circuits for homomorphism polynomials for detecting small induced subgraphs. Brute-force search finds, or even counts $C_6$ as induced subgraphs in an $m$-edge host graph in $O(m^3)$ time. Bl\"aser, Komarath, and Sreenivasaiah \cite{DBLP:conf/fsttcs/BlaserKS18} showed that efficient computation of homomorphism polynomials can be used to speed-up the detection of induced subgraphs. For example, their techniques can be used to show that detecting an induced $C_6$ in an $n$-vertex host graph can be done in $O(n^4)$ time.  In this section, we derive an $\otilde(m^2)$ algorithm for detecting induced $C_6$ in an $m$-edge host graph. This algorithm is a natural analogue of the algorithm by Bla\"ser, Komarath and Sreenivasaiah \cite{DBLP:conf/fsttcs/BlaserKS18}.

We now give some definitions that are necessary to understand how homomorphism polynomials are used in induced subgraph detection.

\begin{definition}\label{def:gind}
We define the induced subgraph isomorphism polynomial for pattern graph $G$ on $n$-vertex host graphs, denoted $\gind_{G}$, as follows. The variables of the polynomial are $y_v$ and $x_{\{u, v\}}$ for all $u, v\in [n]$.
\begin{equation*}
    \gind_{G} = \sum_{G'} \prod_v y_v \prod_e x_e \prod_f (1 - x_f)
\end{equation*}
where $G'$ ranges over all (not-necessarily induced) subgraphs of $K_n$ isomorphic to $G$, $v$ ranges over the vertices of $G'$, $e$ ranges over the edges of $G'$, and $f$ ranges over the edges in $K_n$ between vertices in $G'$ but not in $G'$.
\end{definition}

We denote by $\gind_{G}(E(H))$, the polynomial obtained by substituting the adjacency in $H$ for the edge variables. Note that the monomials of $\gind_{G}(E(H))$ are products of $|V(G)|$ vertex variables and correspond to the induced subgraphs of $H$ isomorphic to $G$. In addition, these monomials all have coefficient 1 because there can only be at most $1$ induced subgraph isomorphic to $G$ on any given $k$ vertices. The induced subgraph polynomials and subgraph polynomials are related via the equation:
\begin{equation}\label{eq:gindgsub}
    \gind_{G}(E(H)) = \sum_{G'} {(-1)}^{|E(G') - E(G)|}\#\gsub_G[G'] \gsub_{G'}[H](x_e = 1) 
\end{equation}
where $G'$ ranges over all $k$-vertex supergraphs of $G$ and $\#\gsub_G[G']$ denotes the number of times $G$ occurs as a subgraph in $G'$. We use the substitution $(x_e = 1)$ to denote that all edge variables in the polynomial are substituted with $1$. Variants of this equation have been used by many authors for induced subgraph detection (See \cite{10.5555/2722129.2722240, DBLP:conf/fsttcs/BlaserKS18,Kowaluk2013CountingAD}).

We now briefly describe how to use homomorphism polynomials to detect induced subgraph isomorphisms (See \cite{DBLP:conf/fsttcs/BlaserKS18} for a more detailed description). Note that the $\gind_G(E(H))$ has the monomial $x_{v_1}\dotsm x_{v_k}$ if and only if $v_1,\dotsc,v_k$ induces a $G$. Therefore, to check whether $G$ occurs as an induced subgraph, we only have to test whether $\gind_G(E(H))$ is non-zero. Furthermore, the coefficient of every monomial is $1$ because a $k$-vertex subgraph can contain at most one induced subgraph isomorphic to $G$. Therefore, whether $H$ contains an induced subgraph isomorphic to $G$ can be reduced to whether $\gind_G(E(H))$ is non-zero \emph{modulo $2$}. The advantage of computing over a ring of characteristic $2$ is that it eliminates all $\gsub_{G'}[H](x_e = 1)$ from the right-hand side of Equation~\ref{eq:gindgsub} for which $\#\gsub_G[G']$ is even. However, we do not have efficient computations for subgraph polynomials. Here, we make use of the observation that $\gsub_{G'}[H](x_e = 1)$ is equal to the multilinear part of $\frac{1}{\#\gaut(G')}\ghom_{G'}[H](x_e = 1)$. Therefore, to test whether $\gsub_{G'}[H](x_e = 1)$ is non-zero modulo $2$, we need only test whether $\frac{1}{\#\gaut(G')}\ghom_{G'}[H](x_e = 1)$ contains a multilinear term with an odd coefficient.

To check the presence of multilinear terms with odd coefficients, we can randomly substitute elements that satisfy the equation $x^2 = 0$ from group algebras over $\mathbb{Z}_2$ \cite{koutis2008faster}. We stress that we only substitute these elements for the vertex variables. The edge variables are all always replaced by $0$ or $1$. The use of a characteristic-2 ring introduces another issue. We now cannot compute $\frac{1}{\#\gaut(G')}\ghom_{G'}[H](x_e = 1)$ for graphs that have an even number of automorphisms, by finding the homomorphism polynomial and dividing by the number of automorphisms. The solution is to compute a polynomial that \emph{avoids} these automorphisms in the multilinear part of $\ghom_{G'}[H](x_e = 1)$ for each such $G'$ so that this division becomes unnecessary, while being careful not to introduce additional multilinear terms. This is the crux of the following proof.

\csixind*

\begin{proof}
We describe how to compute polynomials for which the multilinear part is the same as $\ghom_{G'}[H](x_e = 1)$ \emph{and} the coefficient of all monomials are odd for all $6$-vertex supergraphs of $C_6$ that contain $C_6$ an odd number times. The complete list is given in Figure~\ref{fig:c6ind}.

These computations involve modifying Algorithm~\ref{alg:hom} slightly for each such $G'$. We consider the case of $C_6$. Each multilinear monomial in $\ghom_{C_6}[H](x_e = 1)$ has coefficient $12$. By ensuring that only homomorphisms $\sigma$ where $\sigma(2) = \min(\sigma(2), \sigma(3), \sigma(5), \sigma(6))$ are present in the polynomial, we can ensure that all $C_6$ subgraphs of $H$ are counted exactly thrice, once for each choice of $\{\sigma(1), \sigma(4)\}$. This check can be done when the algorithm processes bag $2365$ (Figure~\ref{fig:c6ind}) in Line~\ref{alg:line22}. Notice that we need to iterate only over the edges present in $H$ in the algorithm. This is because the other edge variables will be substituted with $0$ anyway and those monomials will definitely vanish. This is crucial in ensuring that our algorithm remains $\otilde(m^2)$ and not $\otilde(n^4)$.

We consider one more case from our list. The graph in the first row and third column in Figure~\ref{fig:c6ind} has four automorphisms (horizontal flip and vertical flip). We can ensure that these subgraphs are counted exactly once in the polynomial by ensuring that Line~\ref{alg:line22} in Algorithm~\ref{alg:hom} also checks that $\sigma(3) < \sigma(6)$ (preventing horizontal flips) and $\sigma(2) < \sigma(4)$ (preventing vertical flips). These checks can be done when the algorithm processes the bag $1634$ and $1234$ respectively.

Figure~\ref{fig:c6ind} shows the matched tree decompositions and the constraints on $\sigma$ that can be used to apply Algorithm~\ref{alg:hom} to compute all these polynomials.

\begin{figure}
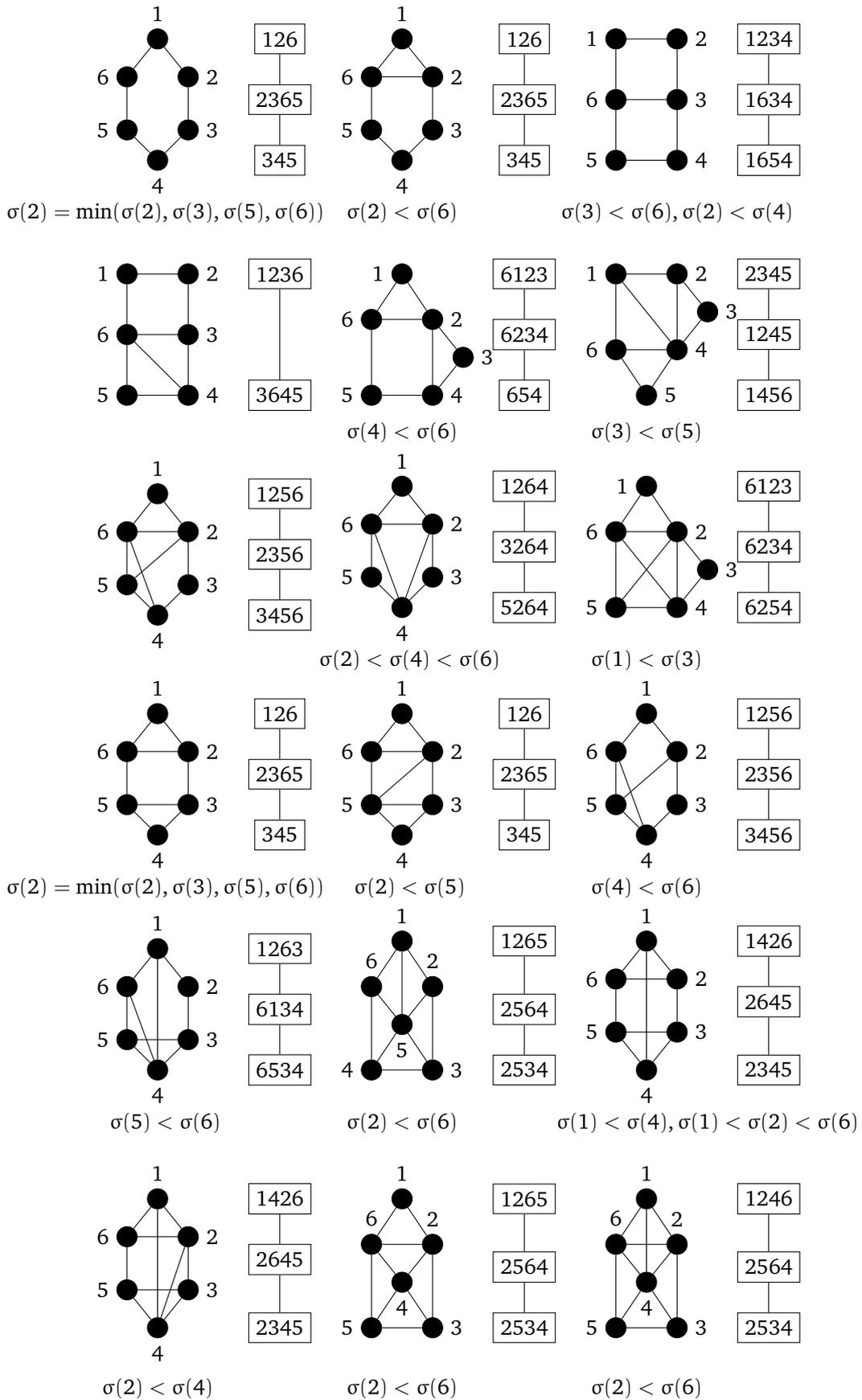

\ctikzfig{c6ind}
\caption{Detecting Induced $C_6$}
\label{fig:c6ind}
\end{figure}
\end{proof}

\pkbarind*

\begin{proof}
We first prove that $\overline{P_k}$ has matched treewidth $k-3$. We only consider the case of odd $k$ (the even case is similar). Let vertices in the path be $1,2,\dotsc,2j+1$. Then, our matched tree decomposition will have three bags. A root bag that excludes only $\{j, j+2\}$, a left child of the root bag that excludes $\{j, j+1\}$, and a right child of the root bag that excludes $\{j+1, j+2\}$. This is clearly a tree decomposition. To see that it is matched, at the root bag, we have the matched matching $(1,2j+1), (2, 2j), \dotsc, (j-1, j+3), (j+1, j+3)$. On the left child, we can keep the rest of edges the same and match $(j+2, j-1)$. Similarly, on the right child, we match $(j, j+3)$.

Since $\overline{P_k}$ has two automorphisms, we also need to show that we can avoid one automorphism while computing the homomorphism polynomial. We note that the non-identity automorphism $\tau$ must have $\tau(1) = k$ and $\tau(k) = 1$. Therefore, we can avoid this by always ensuring that $\sigma(1) < \sigma(k)$ when building the homomorphism polynomial in Algorithm~\ref{alg:hom}.

We know that $\gind_{\overline{P_k}}(E(H)) = \gsub_{\overline{P_k}}[H](x_e = 1) \pmod 2$. The theorem follows.
\end{proof}

\section{Acknowledgement:} The research work of S. Mishra is partially funded by Fondecyt Postdoctoral grant $3220618$ of Agencia National de Investigati\'{o}n y Desarrollo (ANID), Chile.

\bibliography{sparse}
\end{document}